\documentclass{ws-mplb}
\usepackage{amsfonts,amssymb,amsmath}
\usepackage{graphicx}
\usepackage[super,sort,compress]{cite} 
\usepackage{hyperref}
\usepackage{graphicx}
\usepackage{wasysym}
\usepackage{xcolor}

\newcommand{\real}{{\mathbb R}}

\newcommand{\rmcor}{\mathrm{Cor\,}}
\newcommand{\rmre}{\mathrm{Re\,}}

\newcommand{\ascript}{\mathcal{A}}
\newcommand{\bscript}{\mathcal{B}}
\newcommand{\cscript}{\mathcal{C}}
\newcommand{\lscript}{\mathcal{L}}

\newcommand{\atilde}{\widetilde{\ascript}}
\newcommand{\btilde}{\widetilde{\bscript}}
\newcommand{\ctilde}{\widetilde{\cscript}}

\newcommand{\ab}[1]{\left|#1\right|}
\newcommand{\brac}[1]{\left\{#1\right\}}
\newcommand{\paren}[1]{\left(#1\right)}
\newcommand{\sqbrac}[1]{\left[#1\right]}
\newcommand{\elbows}[1]{{\langle#1\rangle}}
\newcommand{\ket}[1]{{|#1\rangle}}
\newcommand{\bra}[1]{{\langle#1|}}

\begin{document}

\markboth{D. D. Georgiev and S. P. Gudder}{Sensitivity of entanglement measures in bipartite pure quantum states}

%
\catchline{}{}{}{}{}
%

\title{Sensitivity of entanglement measures in bipartite pure quantum states}

\author{Danko D. Georgiev}


\address{Institute for Advanced Study, 30 Vasilaki Papadopulu Str., Varna 9010, Bulgaria \\ danko.georgiev@mail.bg}

\author{Stanley P. Gudder}


\address{Department of Mathematics, University of Denver, Denver, CO 80208, USA\\ sgudder@du.edu}

\maketitle

\begin{history}
\received{(Day Month Year)}
\revised{(Day Month Year)}
\end{history}

\begin{abstract}
Entanglement measures quantify the amount of quantum entanglement that is contained in quantum states. Typically, different entanglement measures do not have to be partially ordered. The presence of a definite partial order between two entanglement measures for all quantum states, however, allows for meaningful conceptualization of sensitivity to entanglement, which will be greater for the entanglement measure that produces the larger numerical values. Here, we have investigated the partial order between the normalized versions of four entanglement measures based on Schmidt decomposition of bipartite pure quantum states, namely, concurrence, tangle, entanglement robustness and Schmidt number. We have shown that among those four measures, the concurrence and the Schmidt number have the highest and the lowest sensitivity to quantum entanglement, respectively. Further, we have demonstrated how these measures could be used to track the dynamics of quantum entanglement in a simple quantum toy model composed of two qutrits. Lastly, we have employed state-dependent entanglement statistics to compute measurable correlations between the outcomes of quantum observables in agreement with the uncertainty principle. The presented results could be helpful in quantum applications that require monitoring of the available quantum resources for sharp identification of temporal points of maximal entanglement or system separability.
\end{abstract}


\section{Introduction}

Quantum entanglement is an important resource in quantum information
technologies \cite{Wootters1998a,Vedral2014,Yamamoto2016,Terhal2003}. Shared
quantum entanglement between two distantly located parties allows
for the execution of classically impossible tasks, such as quantum
teleportation \cite{Bennett1993,Pirandola2015}, superdense coding \cite{Bennett1992},
or quantum cryptography \cite{Ekert1991,Bennett1992b,Shenoy2017}. Being such a valuable commodity,
the amount of quantum entanglement possessed by composite quantum
systems has been subject to quantification with a variety of entanglement
measures \cite{Vedral1997,Vedral1998,Horodecki2001,Plenio2007}. Some of these measures were defined operationally \cite{Vedral1997,Eisert2000,Vidal2002b,Horodecki2009},
whereas others were defined with an explicit formula that is computed
from the complex quantum probability amplitudes characterizing the
state of the composite system \cite{Grobe1994,Vidal1999,Hill1997,Wootters1998b,Wootters2001,Rungta2001,Gudder2020a,Gudder2020b}.
The rapid burgeoning of quantum resource theory \cite{Chitambar2019}
has generated a zoo of entanglement measures, most of which appeared
under different names in the works of different authors. This impedes
accessibility of available mathematical results and complicates the
conduction of literature searches. Furthermore, the utility and performance
of different measures for tracking the entanglement dynamics in composite
quantum systems has been rarely compared. To remedy this situation,
in this work we analyze a number of entanglement measures based on
Schmidt decomposition and systematically explore their ability to
resolve maximal entanglement or complete disentanglement of a toy
model system consisting of two interacting qutrits.
Then we provide a comprehensive introduction to state-dependent entanglement statistics
and compute measurable correlations between the outcomes of quantum observables
in agreement with the uncertainty principle.

The organization of the presentation is as follows:
In~Section~\ref{sec:2}, we briefly summarize how every bipartite state vector can be expressed
in the Schmidt basis using singular value decomposition of the complex
coefficient matrix given in some explicit basis. Then, we introduce
four entanglement measures that can be computed directly from the
Schmidt coefficients. The most popular names for these four measures
are: concurrence, tangle, entanglement robustness and Schmidt number.
In~Section~\ref{sec:3}, we introduce the concept of relative sensitivity
to quantum entanglement and prove two main theorems, which establish
the existing partial order between different normalized versions of
the four entanglement measures.
In~Section~\ref{sec:4}, we present a quantum toy model of two interacting qutrits, which ensures
the minimal Hilbert space required to avoid reduction of entanglement robustness to concurrence.
In~Section~\ref{sec:5}, we report computational
results on the performance of each of the four entanglement measures
on resolving maximal entanglement or complete disentanglement of the
toy quantum system.
In~Section~\ref{sec:6}, we introduce the concept of state-dependent entanglement statistics and demonstrate how the Schmidt decomposition features prominently in the computation of measurable correlations between the outcomes of quantum observables.
Finally, we conclude with a brief discussion on the computational
complexity involved in the evaluation of the presented entanglement
measures and their overall utility for tracking the entanglement dynamics
in composite quantum systems.

\section{Entanglement measures}
\label{sec:2}

Quantum entanglement was originally conceptualized by Erwin Schr\"{o}dinger
in the form of probability relations between distant quantum systems
\cite{Schrodinger1935}. It needs to be emphasized, however, that
\emph{quantum probabilities} relate \emph{quantum observables}, which
describe potentialities of what could be measured, but not necessarily
of what is actually measured \cite{Georgiev2018}. This means that
given a quantum state vector $|\Psi\rangle$ of a composite quantum
system, one could always compute the \emph{expectation values} of
different quantum observables, including incompatible (non-commuting)
observables whose simultaneous measurement is physically impossible
\cite{Dirac1967,Susskind2014}. Furthermore, even for maximally entangled quantum
states there are quantum observables whose measurement outcomes are
maximally correlated and quantum observables whose measurement outcomes
are not correlated at all. For example, given the Bell state
$|\Phi^{+}\rangle=\frac{1}{\sqrt{2}} (|\uparrow_{z}\uparrow_{z}\rangle+|\downarrow_{z}\downarrow_{z}\rangle)$
one could either measure the observable $\hat{\sigma}_{z}\otimes\hat{\sigma}_{z}$
obtaining maximally correlated outcomes or measure the observable
$\hat{\sigma}_{z}\otimes\hat{\sigma}_{x}$ obtaining completely uncorrelated
outcomes \cite{Georgiev2021a}. This highlights the fact that quantum
entanglement is not a genuine property of quantum observables. Instead,
the quantum entanglement is a genuine property of the \emph{quantum
state vector} $|\Psi\rangle$, which is comprised of \emph{quantum
probability amplitudes} rather than \emph{quantum probabilities},
and motivates the following definition valid for non-relativistic
quantum mechanics of distinguishable particles.

\begin{definition}
(Entangled state) A bipartite quantum state vector $|\Psi\rangle\in\mathcal{H}_{A}\otimes\mathcal{H}_{B}$
is quantum entangled if and only if it cannot be written as a tensor
product \cite{Gudder2020b} 
\begin{equation}
|\Psi\rangle\neq|\psi\rangle_{A}\otimes|\psi\rangle_{B} .
\end{equation}
Otherwise, the quantum state vector is separable (factorizable).
\end{definition}

The Schmidt decomposition provides a straightforward criterion for
determining whether a bipartite quantum state vector is entangled or not.

\begin{theorem}
(Schmidt decomposition) Consider a composite bipartite quantum state
vector $|\Psi\rangle\in\mathcal{H}_{A}\otimes\mathcal{H}_{B}$. Given
any two complete orthonormal bases for the individual Hilbert spaces,
respectively \textup{$\left\{ |i\rangle_{A}\right\} $} for $\mathcal{H}_{A}$
and \textup{$\left\{ |j\rangle_{B}\right\} $} for $\mathcal{H}_{B}$,
one can always construct a complete orthonormal tensor product basis
$\left\{ |i\rangle_{A}\otimes|j\rangle_{B}\right\} $ for the composite
Hilbert space $\mathcal{H}=\mathcal{H}_{A}\otimes\mathcal{H}_{B}$
in which the bipartite quantum state vector is expressed as \cite{Georgiev2021a}
\begin{equation}
|\Psi\rangle=\sum_{i}\sum_{j}c_{ij}|i\rangle_{A}\otimes|j\rangle_{B} .
\end{equation}
Then, singular value decomposition of the complex coefficient matrix
$\hat{C}=\left(c_{ij}\right)$ renders it in the form 
\begin{equation}
\hat{C}=\hat{U}\hat{\Lambda}\hat{V}^{\dagger} ,
\end{equation}
where $\hat{U}$ and $\hat{V}^{\dagger}$ are unitary matrices, and
$\hat{\Lambda}$ is a diagonal matrix with non-negative singular values
(Schmidt coefficients) sorted in descending order $\lambda_{1}\geq\lambda_{2}\geq\ldots\geq\lambda_{s}\geq0$.
Finally, using the operations of matrix reshaping and reshuffling
in the context of Jamio{\l}kowski isomorphism \cite{Miszczak2011}, the
bipartite quantum state vector can always be expressed in the Schmidt
basis as 
\begin{equation}
|\Psi\rangle=\sum_{s}\lambda_{s}\left(\hat{U}|i_{s}\rangle_{A}\right)\otimes\left(\hat{V}^{\dagger}|j_{s}\rangle_{B}\right) ,
\end{equation}
where the index $s$ runs from 1 to $\min\left[\dim\left(\mathcal{H}_{A}\right),\dim\left(\mathcal{H}_{B}\right)\right]$.
\end{theorem}

\begin{definition}
(Schmidt rank) The number of non-zero Schmidt coefficients is referred
to as the Schmidt rank of a given Schmidt decomposition.
\end{definition}

The Schmidt rank provides a binary, Yes/No, classification of quantum
states. The quantum state is entangled if and only if its Schmidt
rank is greater than~1. For separable states, the Schmidt rank is
exactly~1. Unfortunately, the binary classification of quantum states
does not suffice for quantitative evaluation and management of quantum
entanglement as a resource. Thus, given an entangled state that has
at least two non-zero Schmidt coefficients, it would be useful to
have quantitative measures that determine how valuable the state is.
Next, we present four such entanglement measures whose numerical values
can be computed from explicit formulas involving the Schmidt coefficients, namely, \emph{concurrence},
\emph{tangle}, \emph{entanglement robustness} and \emph{Schmidt number}.

\subsection{Concurrence}

The \emph{concurrence} was first introduced by Hill and Wootters for
pure two qubit states using a modified Bell basis $\left\{ |\Phi^{+}\rangle,\imath|\Phi^{-}\rangle,\imath|\Psi^{+}\rangle,|\Psi^{-}\rangle\right\} $
\cite{Hill1997,Wootters1998b}, but was then generalized as \emph{I-concurrence}
to include multi-level bipartite quantum systems \cite{Rungta2001}
using the sum of the fourth powers of the Schmidt coefficients
\begin{equation}
\mathcal{C}\left(\Psi\right)=\sqrt{2\left[1-\sum_{i=1}^n\lambda_{i}^{4}\right]} ,
\end{equation}
where $n=\min\left[\dim\left(\mathcal{H}_{A}\right),\dim\left(\mathcal{H}_{B}\right)\right]$.

In the context of quantum interferometry with entangled particles, the concurrence is manifested as \emph{two-particle visibility} \cite{Jakob2010,Georgiev2021b,Roy2021}.
Recently, within the context of a general theory of entanglement,
Gudder proposed the \emph{entanglement number}, which is essentially
I-concurrence without the scale factor \cite{Gudder2020a,Gudder2020b}
\begin{equation}
e\left(\Psi\right)=\sqrt{1-\sum_{i}\lambda_{i}^{4}}=\sqrt{1-\textrm{Tr}\left(\hat{\rho}_{A}^{2}\right)}=\sqrt{1-\textrm{Tr}\left(\hat{\rho}_{B}^{2}\right)} .
\label{eq:Gudder-1}
\end{equation}
We can use the fact that the Schmidt decomposition gives a normalized
vector 
\begin{equation}
\sum_{i}\lambda_{i}^{2}=1\label{eq:Schmidt-1}
\end{equation}
in order to substitute in \eqref{eq:Gudder-1} and obtain
\begin{equation}
e\left(\Psi\right)=\sqrt{\sum_{i}\lambda_{i}^{2}-\sum_{j}\lambda_{j}^{4}}=\sqrt{\sum_{i}\left(\lambda_{i}^{2}-\lambda_{i}^{4}\right)}=\sqrt{\sum_{i}\lambda_{i}^{2}\left(1-\lambda_{i}^{2}\right)} .
\end{equation}
Since \eqref{eq:Schmidt-1} implies that 
\begin{equation}
1-\lambda_{i}^{2}=\sum_{j\neq i}\lambda_{j}^{2} ,
\end{equation}
one arrives at
\begin{equation}
e\left(\Psi\right)=\sqrt{\sum_{i\neq j}\lambda_{i}^{2}\lambda_{j}^{2}}=\sqrt{2\sum_{i<j}\lambda_{i}^{2}\lambda_{j}^{2}} .
\end{equation}
Alternatively, it is possible to use the identity $1^{2}=1$ to directly
obtain
\begin{equation}
e\left(\Psi\right)=\sqrt{1^{2}-\sum_{j}\lambda_{j}^{4}}=\sqrt{\left(\sum_{i}\lambda_{i}^{2}\right)^{2}-\sum_{j}\lambda_{j}^{4}}=\sqrt{\sum_{i\neq j}\lambda_{i}^{2}\lambda_{j}^{2}} .
\end{equation}
The range of the entanglement number is within the interval $0\leq e(\Psi)\leq\sqrt{\frac{n-1}{n}}$.

The normalized concurrence is the same as the normalized entanglement number \cite{Georgiev2021a}
\begin{equation}
\tilde{\mathcal{C}}\left(\Psi\right)=\tilde{e}\left(\Psi\right)=\sqrt{\frac{n}{n-1}\left(1-\sum_{i}\lambda_{i}^{4}\right)} .
\label{eq:norm-concurrence}
\end{equation}

Computationally useful is the fact that the entanglement number
and concurrence could be evaluated from squaring the Hermitian matrix $\hat{C}\hat{C}^{\dagger}$ obtained from the complex coefficient
matrix $\hat{C}=\left(c_{ij}\right)$ without the need of singular
value decomposition \cite{Gudder2020a,Gudder2020b}
\begin{equation}
e\left(\Psi\right)=\sqrt{\left[\textrm{Tr}\left(\hat{C}\hat{C}^{\dagger}\right)\right]^{2}-\textrm{Tr}\left[\left(\hat{C}\hat{C}^{\dagger}\right)^{2}\right]}=\sqrt{1-\textrm{Tr}\left[\left(\hat{C}\hat{C}^{\dagger}\right)^{2}\right]} .
\end{equation}

\subsection{Tangle}

The \emph{squared concurrence} $\mathcal{C}^{2}\left(\Psi\right)$
is a distinct entanglement measure referred to as \emph{tangle} \cite{Wootters1998a,Rungta2003}.
From \eqref{eq:norm-concurrence}, the normalized tangle is given
by
\begin{equation}
\tilde{\mathcal{T}}\left(\Psi\right)=\tilde{\mathcal{C}}^{2}\left(\Psi\right)=\frac{n}{n-1}\left(1-\sum_{i}\lambda_{i}^{4}\right)=\frac{n}{n-1}\sum_{i\neq j}\lambda_{i}^{2}\lambda_{j}^{2} .
\label{eq:tangle}
\end{equation}

\subsection{Robustness of entanglement}

For pure bipartite states, the \emph{robustness of entanglement} \cite{Vidal1999,Krynytskyi2021}
can be computed from the squared sum of the Schmidt coefficients as
follows
\begin{equation}
\mathcal{R}\left(\Psi\right)=\left(\sum_{i=1}^{n}\lambda_{i}\right)^{2}-1 ,
\end{equation}
where $n=\min\left[\dim\left(\mathcal{H}_{A}\right),\dim\left(\mathcal{H}_{B}\right)\right]$.

The range of the robustness of entanglement is 
\begin{equation}
0\leq\mathcal{R}(\Psi)\leq n-1 .
\end{equation}
Therefore, the normalized robustness $0\leq\tilde{\mathcal{R}}(\Psi)\leq1$
is given by
\begin{equation}
\tilde{\mathcal{R}}(\Psi)=\frac{1}{n-1}\left[\left(\sum_{i}\lambda_{i}\right)^{2}-1\right] .
\label{eq:R}
\end{equation}

Straightforward algebraic calculation shows that the robustness of
entanglement is the same as the non-normalized coherence \cite{Bera2015}
of the density matrix in the Schmidt basis 
\begin{align}
\mathcal{R}\left(\Psi\right) & =\left(\sum_{i}\lambda_{i}\right)^{2}-1=\left(\sum_{i}\lambda_{i}\right)^{2}-\sum_{i}\lambda_{i}^{2}\nonumber \\
 & =\left(\sum_{i}\lambda_{i}^{2}+2\sum_{i<j}\lambda_{i}\lambda_{j}\right)-\sum_{i}\lambda_{i}^{2}\nonumber \\
 & =2\sum_{i<j}\lambda_{i}\lambda_{j}=\sum_{i\neq j}\lambda_{i}\lambda_{j} .
\end{align}

The normalized robustness of entanglement is then the same as the
\emph{entanglement coherence} $C_{E}$ defined in the Schmidt basis
\cite{Pathania2020} 
\begin{equation}
\tilde{\mathcal{R}}\left(\Psi\right)=C_{E}=\frac{1}{n-1}\sum_{i\neq j}\lambda_{i}\lambda_{j} .
\label{eq:robustness}
\end{equation}

The entanglement coherence could be also expressed in terms of the
reduced density matrices $\hat{\rho}_{A}=\textrm{Tr}_{B}\left(\hat{\rho}_{AB}\right)$
and $\hat{\rho}_{B}=\textrm{Tr}_{A}\left(\hat{\rho}_{AB}\right)$
as follows 
\begin{equation}
C_{E}=\frac{1}{n-1}\left[\left(\textrm{Tr}\sqrt{\hat{\rho}_{A}}\right)^{2}-1\right]=\frac{1}{n-1}\left[\left(\textrm{Tr}\sqrt{\hat{\rho}_{B}}\right)^{2}-1\right] .
\end{equation}

In the special case of two qubits, the robustness of entanglement reduces to concurrence
\begin{equation}
\tilde{\mathcal{C}}\left(\Psi\right)=2\lambda_{1}\lambda_{2}=\tilde{\mathcal{R}}\left(\Psi\right) ,
\end{equation}
but for higher dimensional systems that admit more than two non-zero
Schmidt coefficients, those two entanglement measures are different.

\subsection{Schmidt number}

The \emph{Schmidt number} \cite{Law2005,Bogdanov2007,Kanada2015,Fastovets2021},
also referred to as \emph{degree of correlation} \cite{Grobe1994},
is another entanglement measure that uses the sum of the fourth powers
of the Schmidt coefficients 
\begin{equation}
\mathcal{K}\left(\Psi\right)=\frac{1}{{\displaystyle \sum_i}\lambda_{i}^{4}} .
\label{eq:K-non}
\end{equation}
The Schmidt number could be interpreted as counting the average number
of Schmidt modes actively involved in entanglement \cite{Law2005}.
The range of the Schmidt number is within the interval $1\leq\mathcal{K}\left(\Psi\right)\leq n$.

Therefore, the normalized Schmidt number $0\leq\tilde{\mathcal{K}}\left(\Psi\right)\leq1$
is given by 
\begin{equation}
\tilde{\mathcal{K}}\left(\Psi\right)=\frac{1}{n-1}\left(\frac{1}{{\displaystyle \sum_{i}}\lambda_{i}^{4}}-1\right) .
\label{eq:K}
\end{equation}

\section{Relative sensitivity to quantum entanglement}
\label{sec:3}

\begin{definition}
(Relative sensitivity to quantum entanglement) Given two normalized
entanglement measures $0\leq\tilde{\mathcal{A}}(\Psi)\leq1$ and $0\leq\tilde{\mathcal{B}}(\Psi)\leq1$,
we say that $\tilde{\mathcal{B}}(\Psi)$ is more sensitive to quantum
entanglement compared with $\tilde{\mathcal{A}}(\Psi)$ if and only
if the ordering $\tilde{\mathcal{A}}(\Psi)\leq\tilde{\mathcal{B}}(\Psi)$
holds for any state $\ket{\Psi}$. Otherwise, we say that the two measures
are unordered and their relative sensitivity is undefined.
\end{definition}

\begin{theorem}
The normalized versions of the Schmidt number $\tilde{\mathcal{K}}\left(\Psi\right)$,
tangle $\mathcal{\tilde{C}}^{2}\left(\Psi\right)$ and concurrence
$\mathcal{\tilde{C}}\left(\Psi\right)$ are ordered in an increasing
order of sensitivity to quantum entanglement, namely, for any state
$\ket{\Psi}$ whose singular value decomposition is described by a set of
Schmidt coefficients $\left\{ \lambda_{i}\right\} _{i=1}^{n}$, we
have
\begin{equation}
\tilde{\mathcal{K}}\left(\Psi\right)\leq\mathcal{\tilde{C}}^{2}\left(\Psi\right)\leq\mathcal{\tilde{C}}\left(\Psi\right) .
\end{equation}
\end{theorem}

\begin{proof}
The relationship $\mathcal{\tilde{C}}\left(\Psi\right)\geq\mathcal{\tilde{C}}^{2}\left(\Psi\right)$
between concurrence and tangle is straightforward and follows from
the fact that the concurrence is bounded within $[0,1]$, namely,
$\mathcal{\tilde{C}}\left(\Psi\right)\leq1$. Therefore, $\mathcal{\tilde{C}}\left(\Psi\right)\times\mathcal{\tilde{C}}\left(\Psi\right)\leq\mathcal{\tilde{C}}\left(\Psi\right)\times1$.

To show that $\mathcal{\tilde{C}}^{2}\left(\Psi\right)\geq\tilde{\mathcal{K}}\left(\Psi\right)$,
we factor the normalized Schmidt number \eqref{eq:K} as follows
\begin{equation}
\tilde{\mathcal{K}}\left(\Psi\right)
=\frac{1}{n-1}\left(\frac{1-\sum_{i}\lambda_{i}^{4}}{\sum_{i}\lambda_{i}^{4}}\right)=\frac{1}{n-1}\left(1-\sum_{i}\lambda_{i}^{4}\right)\left(\frac{1}{\sum_{i}\lambda_{i}^{4}}\right) .
\label{eq:K-factor}
\end{equation}
Because the quantity $\sum_{i}\lambda_{i}^{4}$ has a minimum only
when there are $n$ equal Schmidt coefficients, each with value of
$\frac{1}{\sqrt{n}}$, we obtain 
\begin{equation}
\sum_{i}\lambda_{i}^{4}\geq n\left(\frac{1}{\sqrt{n}}\right)^{4}=\frac{1}{n} .
\label{eq:sum-4}
\end{equation}
Taking the reciprocal values gives
\begin{equation}
\frac{1}{\sum_{i}\lambda_{i}^{4}}\leq n .
\label{eq:inv-sum-4}
\end{equation}
This inequality can also be directly proved using Sedrakyan's inequality
\cite{Sedrakyan2018} as follows
\begin{equation*}
\frac{1^{2}}{\sum_{i}\lambda_{i}^{4}}=\frac{\left(\sum_{i}\lambda_{i}^{2}\right)^{2}}{\sum_{i}\lambda_{i}^{4}}=\frac{\left(\lambda_{1}^{2}+\lambda_{2}^{2}+\ldots+\lambda_{n}^{2}\right)^{2}}{\lambda_{1}^{4}+\lambda_{2}^{4}+\ldots+\lambda_{n}^{4}}\leq\frac{\left(\lambda_{1}^{2}\right)^{2}}{\lambda_{1}^{4}}+\frac{\left(\lambda_{2}^{2}\right)^{2}}{\lambda_{2}^{4}}+\ldots+\frac{\left(\lambda_{n}^{2}\right)^{2}}{\lambda_{n}^{4}}=1\times n .
\end{equation*}
After substitution of \eqref{eq:inv-sum-4} in the Schmidt number
\eqref{eq:K-factor}, we conclude
\begin{equation}
\tilde{\mathcal{K}}\left(\Psi\right)=\frac{1}{n-1}\left(1-\sum_{i}\lambda_{i}^{4}\right)\left(\frac{1}{\sum_{i}\lambda_{i}^{4}}\right)\leq\frac{1}{n-1}\left(1-\sum_{i}\lambda_{i}^{4}\right)\times n=\mathcal{\tilde{C}}^{2}\left(\Psi\right) .
\end{equation}
This establishes the chain of inequalities as stated in the theorem.
\end{proof}

\begin{theorem}
The normalized versions of the Schmidt number $\tilde{\mathcal{K}}\left(\Psi\right)$,
robustness $\mathcal{\tilde{R}}\left(\Psi\right)$ and concurrence
$\mathcal{\tilde{C}}\left(\Psi\right)$ are ordered in an increasing
order of sensitivity to entanglement, namely, for the same set of
Schmidt coefficients we have
\begin{equation*}
\tilde{\mathcal{K}}\left(\Psi\right)\leq\mathcal{\tilde{R}}\left(\Psi\right)\leq\mathcal{\tilde{C}}\left(\Psi\right) .
\end{equation*}
\end{theorem}

\begin{proof}
To show that $\mathcal{\tilde{C}}\left(\Psi\right)\geq\mathcal{\tilde{R}}\left(\Psi\right)$,
we use \eqref{eq:robustness} to compute the square of the robustness
\begin{equation*}
\tilde{\mathcal{R}}^{2}\left(\Psi\right)=\frac{1}{\left(n-1\right)^{2}}\left(2\sum_{i<j}\lambda_{i}\lambda_{j}\right)^{2}
\end{equation*}
and then compare it to the tangle
\begin{equation*}
\tilde{\mathcal{C}}^{2}\left(\Psi\right)=\frac{n}{n-1}\times2\sum_{i<j}\left(\lambda_{i}\lambda_{j}\right)^{2} .
\end{equation*}

For $n=2$, we have $\mathcal{\tilde{R}}\left(\Psi\right)=\mathcal{\tilde{C}}\left(\Psi\right)=2\lambda_{1}\lambda_{2}$.

For $n\geq3$, we will have $n$ Schmidt coefficients, $\lambda_{1},\lambda_{2},\ldots,\lambda_{n}\geq0$,
which will generate $N=\left(\begin{array}{c}
n\\
2
\end{array}\right)=\frac{n\left(n-1\right)}{2}$ pairs $\lambda_{i}\lambda_{j}$ for which $i<j$. After re-indexing
all $N$ pairs using the single index
\begin{equation}
k=\left(i-1\right)n-\left(\begin{array}{c}
i\\
2
\end{array}\right)+j-i=\frac{1}{2}\left(i-1\right)\left(2n-i\right)+j-i
\end{equation}
 and introducing the combined variable $x_{k}=\lambda_{i}\lambda_{j}\geq0$,
we have
\begin{equation*}
\frac{\left(n-1\right)^{2}}{4}\tilde{\mathcal{R}}^{2}\left(\Psi\right)=\left(\sum_{i<j}\lambda_{i}\lambda_{j}\right)^{2}=\left(\sum_{k=1}^{N}x_{k}\right)^{2}=\sum_{k}x_{k}^{2}+2\sum_{k<l}x_{k}x_{l} .
\end{equation*}
Similarly, we obtain
\begin{equation*}
\frac{\left(n-1\right)^{2}}{4}\tilde{\mathcal{C}}^{2}\left(\Psi\right)=\frac{n\left(n-1\right)}{2}\sum_{i<j}\left(\lambda_{i}\lambda_{j}\right)^{2}=N\times\sum_{k=1}^{N}x_{k}^{2}=\sum_{k}x_{k}^{2}+\sum_{k<l}\left(x_{k}^{2}+x_{l}^{2}\right) .
\end{equation*}
Now, we take the difference
\begin{align*}
\frac{\left(n-1\right)^{2}}{4}\left[\tilde{\mathcal{C}}^{2}\left(\Psi\right)-\tilde{\mathcal{R}}^{2}\left(\Psi\right)\right] & =\sum_{k}x_{k}^{2}+\sum_{k<l}\left(x_{k}^{2}+x_{l}^{2}\right)-\sum_{k}x_{k}^{2}-2\sum_{k<l}x_{k}x_{l}\\
 & =\sum_{k<l}\left(x_{k}-x_{l}\right)^{2}\geq0 .
\end{align*}
Since $\frac{\left(n-1\right)^{2}}{4}>0$, it follows that $\tilde{\mathcal{C}}\left(\Psi\right)\geq\tilde{\mathcal{R}}\left(\Psi\right)$.

To show that $\tilde{\mathcal{R}}\left(\Psi\right)\geq\tilde{\mathcal{K}}\left(\Psi\right)$,
we subtract \eqref{eq:K} from \eqref{eq:R} and observe that the
relation $\left[\tilde{\mathcal{R}}\left(\Psi\right)-\tilde{\mathcal{K}}\left(\Psi\right)\right]\geq0$
is equivalent to 
\begin{equation}
\left(\sum_{i}\lambda_{i}\right)^{2}\geq\frac{1}{\sum_{i}\lambda_{i}^{4}} .
\label{eq:triple-0}
\end{equation}
The general proof for $n\geq2$ of the later relation is quite intricate
due to the presence of the sum of fourth powers in the denominator
of the fraction. Combining the normalization of the Schmidt coefficients
$\lambda_{1}^{2}+\lambda_{2}^{2}+\ldots+\lambda_{n}^{2}=1$ with $1^{3}=1$
transforms \eqref{eq:triple-0} into
\begin{equation}
A\equiv\left(\sum_{i}\lambda_{i}\right)^{2}\left(\sum_{i}\lambda_{i}^{4}\right)\geq\left(\sum_{i}\lambda_{i}^{2}\right)^{3}\equiv B .
\label{eq:triple-1}
\end{equation}
To proceed, we will need the following factorization
\begin{equation}
\lambda_{i}^{4}+\lambda_{j}^{4}-\lambda_{i}\lambda_{j}\left(\lambda_{i}^{2}+\lambda_{j}^{2}\right)=\left(\lambda_{i}-\lambda_{j}\right)^{2}\left(\lambda_{i}^{2}+\lambda_{i}\lambda_{j}+\lambda_{j}^{2}\right)\geq0
\end{equation}
and the inequality between arithmetic mean and geometric mean for
a triple of non-negative numbers rearranged in the form
\begin{equation}
\lambda_{i}^{3}+\lambda_{j}^{3}+\lambda_{k}^{3}-3\lambda_{i}\lambda_{j}\lambda_{k}\geq0 .
\end{equation}
Next, we expand the two sides of \eqref{eq:triple-1} into groups
of similar terms
\begin{align}
A &=  \sum_{i}\lambda_{i}^{6}
+2\sum_{i<j}\lambda_{i}\lambda_{j}\left(\lambda_{i}^{4}+\lambda_{j}^{4}\right)
+\sum_{i<j}\lambda_{i}^{2}\lambda_{j}^{2}\left(\lambda_{i}^{2}+\lambda_{j}^{2}\right)\nonumber\\
& \quad +2\sum_{i<j<k}\lambda_{i}\lambda_{j}\lambda_{k}\left(\lambda_{i}^{3}+\lambda_{j}^{3}+\lambda_{k}^{3}\right) ,
\label{eq:triple-2} \\
B &=\sum_{i}\lambda_{i}^{6}+3\sum_{i<j}\lambda_{i}^{2}\lambda_{j}^{2}\left(\lambda_{i}^{2}+\lambda_{j}^{2}\right)+6\sum_{i<j<k}\lambda_{i}^{2}\lambda_{j}^{2}\lambda_{k}^{2} .
\label{eq:triple-3}
\end{align}
Taking the difference between \eqref{eq:triple-2} and \eqref{eq:triple-3} gives
\begin{align}
\frac{A-B}{2} & =\sum_{i<j}\lambda_{i}\lambda_{j}\left(\lambda_{i}^{4}+\lambda_{j}^{4}\right)
-\sum_{i<j}\lambda_{i}^{2}\lambda_{j}^{2}\left(\lambda_{i}^{2}+\lambda_{j}^{2}\right) \nonumber\\
& \quad +\sum_{i<j<k}\lambda_{i}\lambda_{j}\lambda_{k}\left(\lambda_{i}^{3}+\lambda_{j}^{3}+\lambda_{k}^{3}\right)
-3\sum_{i<j<k}\lambda_{i}^{2}\lambda_{j}^{2}\lambda_{k}^{2}\nonumber \\
 & =\sum_{i<j}\lambda_{i}\lambda_{j}\left(\lambda_{i}-\lambda_{j}\right)^{2}\left(\lambda_{i}^{2}+\lambda_{i}\lambda_{j}+\lambda_{j}^{2}\right) \nonumber\\
& \quad +\sum_{i<j<k}\lambda_{i}\lambda_{j}\lambda_{k}\left(\lambda_{i}^{3}+\lambda_{j}^{3}+\lambda_{k}^{3}-3\lambda_{i}\lambda_{j}\lambda_{k}\right)\nonumber\\
& \geq0 .
\label{eq:triple-4}
\end{align}
In the last step, we have used the fact that the Schmidt coefficients
are non-negative, $\lambda_{i}\geq0$. The minimal dimensional case,
$n=2$, is obtained trivially from \eqref{eq:triple-4} by plugging
in $\lambda_{3}=0$. Thus, the Schmidt number is a lower bound on
robustness, namely, $\left[\tilde{\mathcal{R}}\left(\Psi\right)-\tilde{\mathcal{K}}\left(\Psi\right)\right]\geq0$.
\end{proof}

\section{Quantum toy model}
\label{sec:4}

The relative sensitivity to quantum entanglement could be exploited
in quantum applications, which require sharp resolution of either
maximal entanglement or complete separability of bipartite quantum
systems. In such cases, the choice of suitable entanglement measure
could depend on the actual practical task. For example, to achieve
the sharpest possible resolution of maximally entangled states, one
could use the Schmidt number $\tilde{\mathcal{K}}\left(\Psi\right)$,
which has lowest sensitivity to quantum entanglement. Conversely,
to achieve the sharpest possible resolution of completely separable
states, one could use the concurrence $\tilde{\mathcal{C}}\left(\Psi\right)$,
which has highest sensitivity to quantum entanglement. To illustrate
the content of the latter two statements, next we construct and employ
a minimal quantum toy model.

A major simplification of the toy model can be accomplished by noting
that quantum dynamics due to internal Hamiltonians does not have an
impact on the entanglement contained in a composite system \cite{Georgiev2021a}.
Indeed, suppose that the initial state of a bipartite quantum system
is expressed in the Schmidt basis as
\begin{equation}
|\Psi(0)\rangle=\sum_{s}\lambda_{s}|i\rangle_{A}\otimes|j\rangle_{B} .
\end{equation}
If $\hat{H}_{A}$ is the internal Hamiltonian of subsystem $A$ and
$\hat{H}_{B}$ is the internal Hamiltonian of subsystem $B$, the
overall unitary action with the Hamiltonian $\hat{H}_{A}\otimes\hat{I}_{B}+\hat{I}_{A}\otimes\hat{H}_{B}$
results in
\begin{align}
|\Psi(t)\rangle & =e^{-\frac{\imath}{\hbar}\left(\hat{H}_{A}\otimes\hat{I}_{B}+\hat{I}_{A}\otimes\hat{H}_{B}\right)t}|\Psi(0)\rangle\nonumber \\
 & =\sum_{s}\lambda_{s}\left(e^{-\frac{\imath}{\hbar}\hat{H}_{A}t}|i\rangle_{A}\right)\otimes\left(e^{-\frac{\imath}{\hbar}\hat{H}_{B}t}|j\rangle_{B}\right)\nonumber \\
 & =\sum_{s}\lambda_{s}|i(t)\rangle_{A}\otimes|j(t)\rangle_{B} .
\end{align}
Because the individual operators $e^{-\frac{\imath}{\hbar}\hat{H}_{A}t}$
and $e^{-\frac{\imath}{\hbar}\hat{H}_{B}t}$ are also unitary, they
preserve the orthonormality of the basis sets $\left\{ |i(t)\rangle_{A}\right\} $
and $\left\{ |j(t)\rangle_{B}\right\} $ at all time $t$. This means
that the Schmidt coefficients $\left\{ \lambda_{s}\right\} $ remain
constant (do not evolve in time) and any entanglement measure dependent
on the Schmidt coefficients remains constant too \cite{Georgiev2021a}.
The latter result allows us to set without loss of generality, $\hat{H}_{A}=0$
and $\hat{H}_{B}=0$, investigating the quantum dynamics resulting
solely due to a non-zero interaction Hamiltonian $\hat{H}_{\textrm{int}}\neq0$.
In fact, it is exactly the dynamics due to the non-zero interaction
Hamiltonian that has the capacity to entangle or disentangle the composite
system \cite{Georgiev2021a}.

The minimal bipartite quantum model that prevents reduction of concurrence
to the robustness of entanglement requires the composition of at least
three-level subsystems (qutrits). Therefore, for the construction
of our quantum toy model we choose two qutrits governed by the spin-1
Heisenberg interaction Hamiltonian \cite{Piddock2021}
\begin{equation}
\hat{H}_{\textrm{int}}=\hbar\omega\left(\hat{\sigma}_{x}\otimes\hat{\sigma}_{x}+\hat{\sigma}_{y}\otimes\hat{\sigma}_{y}+\hat{\sigma}_{z}\otimes\hat{\sigma}_{z}\right) ,
\end{equation}
where the spin-1 matrices are given by 
\begin{align}
\hat{\sigma}_{x} & =\frac{1}{\sqrt{2}}\left(\begin{array}{ccc}
0 & 1 & 0\\
1 & 0 & 1\\
0 & 1 & 0
\end{array}\right) ,\\
\hat{\sigma}_{y} & =\frac{1}{\sqrt{2}}\imath\left(\begin{array}{ccc}
0 & -1 & 0\\
1 & 0 & -1\\
0 & 1 & 0
\end{array}\right) ,\\
\hat{\sigma}_{z} & =\left(\begin{array}{ccc}
1 & 0 & 0\\
0 & 0 & 0\\
0 & 0 & -1
\end{array}\right) .
\end{align}
The eigenvectors of the observable $\hat{\sigma}_{z}$ form the basis
set $\left\{ |\uparrow\rangle,|\ocircle\rangle,|\downarrow\rangle\right\} $,
respectively with eigenvalues $\left\{ 1,0,-1\right\} $.

To solve the Schr\"{o}dinger equation for any initial state $|\Psi(0)\rangle$,
it would be convenient to follow the standard procedure utilizing
the energy eigenbasis. The interaction Hamiltonian
$\hat{H}_{\textrm{int}}$ has five eigenstates with eigenvalue $\hbar\omega$:
\begin{align}
|E_{1}\rangle & =|\uparrow\uparrow\rangle ,\label{eq:E1}\\
|E_{2}\rangle & =\frac{1}{\sqrt{2}}\left(|\uparrow\ocircle\rangle+|\ocircle\uparrow\rangle\right) ,\label{eq:E2}\\
|E_{3}\rangle & =\frac{1}{\sqrt{6}}\left(|\uparrow\downarrow\rangle+2|\ocircle\ocircle\rangle+|\downarrow\uparrow\rangle\right) ,\label{eq:E3}\\
|E_{4}\rangle & =\frac{1}{\sqrt{2}}\left(|\ocircle\downarrow\rangle+|\downarrow\ocircle\rangle\right) ,\label{eq:E4}\\
|E_{5}\rangle & =|\downarrow\downarrow\rangle ,\label{eq:E5}
\end{align}
three eigenstates with eigenvalue $-\hbar\omega$:
\begin{align}
|E_{6}\rangle & =-\frac{1}{\sqrt{2}}\left(|\uparrow\ocircle\rangle-|\ocircle\uparrow\rangle\right) ,\label{eq:E6}\\
|E_{7}\rangle & =-\frac{1}{\sqrt{2}}\left(|\uparrow\downarrow\rangle-|\downarrow\uparrow\rangle\right) ,\label{eq:E7}\\
|E_{8}\rangle & =-\frac{1}{\sqrt{2}}\left(|\ocircle\downarrow\rangle-|\downarrow\ocircle\rangle\right) ,\label{eq:E8}
\end{align}
and a single eigenstate with eigenvalue $-2\hbar\omega$:
\begin{equation}
|E_{9}\rangle=\frac{1}{\sqrt{3}}\left(|\uparrow\downarrow\rangle-|\ocircle\ocircle\rangle+|\downarrow\uparrow\rangle\right) .\label{eq:E9}
\end{equation}
Throughout this work, we will set $\omega=1$ rad/ps in the interaction
Hamiltonian between the two qutrits.

The general solution of the Schr\"{o}dinger equation in the energy eigenbasis
is then 
\begin{equation}
\imath\hbar\frac{\partial}{\partial t}|\Psi\rangle=\hat{H}|\Psi\rangle=\hat{H}\sum_{n}\alpha_{n}|E_{n}\rangle=\sum_{n}E_{n}\alpha_{n}|E_{n}\rangle ,
\end{equation}
which can be explicitly written as
\begin{align}
|\Psi(t)\rangle & =\sum_{n}\alpha_{n}e^{-\frac{\imath}{\hbar}E_{n}t}|E_{n}\rangle\nonumber \\
 & =e^{-\imath\omega t}\left(\alpha_{1}|E_{1}\rangle+\alpha_{2}|E_{2}\rangle+\alpha_{3}|E_{3}\rangle+\alpha_{4}|E_{4}\rangle+\alpha_{5}|E_{5}\rangle\right)\nonumber \\
 & \quad+e^{\imath\omega t}\left(\alpha_{6}|E_{6}\rangle+\alpha_{7}|E_{7}\rangle+\alpha_{8}|E_{8}\rangle\right)+e^{2\imath\omega t}\alpha_{9}|E_{9}\rangle
, \label{eq:psi-t}
\end{align}
where $\alpha_{n}$ is the initial quantum probability amplitude of
the state $|E_{n}\rangle$ at $t=0$.

\section{Quantum dynamics of entanglement measures}
\label{sec:5}

Dynamics of the expectation value of any quantum observable $\hat{A}$
could be obtained from the solution \eqref{eq:psi-t} together with
the Born rule
\begin{equation}
\langle\hat{A}\rangle=\langle\Psi(t)|\hat{A}|\Psi(t)\rangle
\end{equation}
Here, we have chosen to track as quantum observables the individual
projectors onto the eigenstates of $\hat{\sigma}_{z}\otimes\hat{\sigma}_{z}$,
namely, $\hat{\mathcal{P}}\left(\uparrow\uparrow\right)=|\uparrow\uparrow\rangle\langle\uparrow\uparrow|$,
$\hat{\mathcal{P}}\left(\uparrow\ocircle\right)=|\uparrow\ocircle\rangle\langle\uparrow\ocircle|$,
$\hat{\mathcal{P}}\left(\uparrow\downarrow\right)=|\uparrow\downarrow\rangle\langle\uparrow\downarrow|$,
$\hat{\mathcal{P}}\left(\ocircle\uparrow\right)=|\ocircle\uparrow\rangle\langle\ocircle\uparrow|$,
$\hat{\mathcal{P}}\left(\ocircle\ocircle\right)=|\ocircle\ocircle\rangle\langle\ocircle\ocircle|$,
$\hat{\mathcal{P}}\left(\ocircle\downarrow\right)=|\ocircle\downarrow\rangle\langle\ocircle\downarrow|$,
$\hat{\mathcal{P}}\left(\downarrow\uparrow\right)=|\downarrow\uparrow\rangle\langle\downarrow\uparrow|$,
$\hat{\mathcal{P}}\left(\downarrow\ocircle\right)=|\downarrow\ocircle\rangle\langle\downarrow\ocircle|$
and $\hat{\mathcal{P}}\left(\downarrow\downarrow\right)=|\downarrow\downarrow\rangle\langle\downarrow\downarrow|$.

For the initial state $|\Psi(0)\rangle$ we considered each of the
nine eigenstates of $\hat{\sigma}_{z}\otimes\hat{\sigma}_{z}$. Projecting
each of these initial states onto the energy eigenbasis using \eqref{eq:E1}-\eqref{eq:E9}
gives the initial values for the energy quantum probability amplitudes
$\alpha_{n}$ for the corresponding simulations. Due to the existing
freedom of choice for the $\uparrow$ vs $\downarrow$ direction in
space, there are mirror symmetries in the obtained solutions, which
can be grouped into 4 cases.

\subsection{Case 0}

Trivial quantum dynamics, in which the initial state is also an energy
eigenstate, occurs when $|\Psi(0)\rangle=|\uparrow\uparrow\rangle=|E_{1}\rangle$
or $|\Psi(0)\rangle=|\downarrow\downarrow\rangle=|E_{5}\rangle$.
In those situations, the expectation values for the respective projectors
remain unitary at all times, namely $\langle\hat{\mathcal{P}}\left(\uparrow\uparrow\right)\rangle=1$
or $\langle\hat{\mathcal{P}}\left(\downarrow\downarrow\right)\rangle=1$.
The composite state vector $|\Psi(t)\rangle$ remains in a separable
state at all times with constant Schmidt coefficients
\begin{equation}
\begin{cases}
\lambda_{1} & =1\\
\lambda_{2} & =0\\
\lambda_{3} & =0 .
\end{cases} 
\end{equation}
All four entanglement measures remain zero at all times.

\subsection{Case 1}

\begin{figure}[t]
\begin{centering}
\includegraphics[width=140mm]{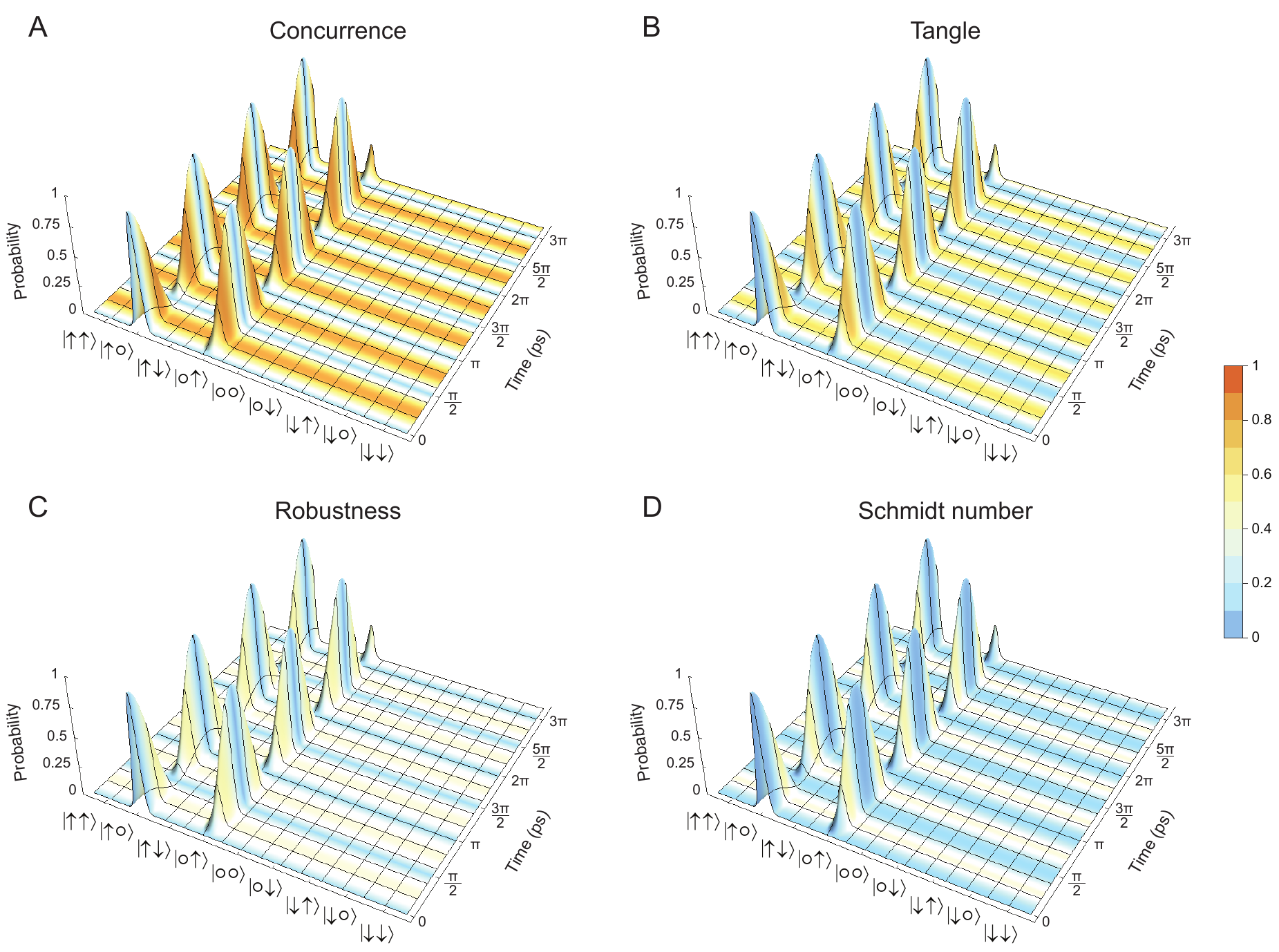}
\par\end{centering}

\caption{\label{fig:1}Dynamics of the expectation values for the respective
projectors onto the eigenstates of the quantum observable $\hat{\sigma}_{z}\otimes\hat{\sigma}_{z}$
simulated with the initial state $|\Psi(0)\rangle=|\uparrow\ocircle\rangle$.
The amount of quantum entanglement is measured with the use of concurrence
in panel (A), tangle in panel (B), robustness of entanglement in panel
(C) and the Schmidt number in panel (D). The coupling strength $\hbar\omega$
in the interaction Hamiltonian between the two qutrits was set by
$\omega=1$ rad/ps.}
\end{figure}

\begin{figure}[t]
\begin{centering}
\includegraphics[width=120mm]{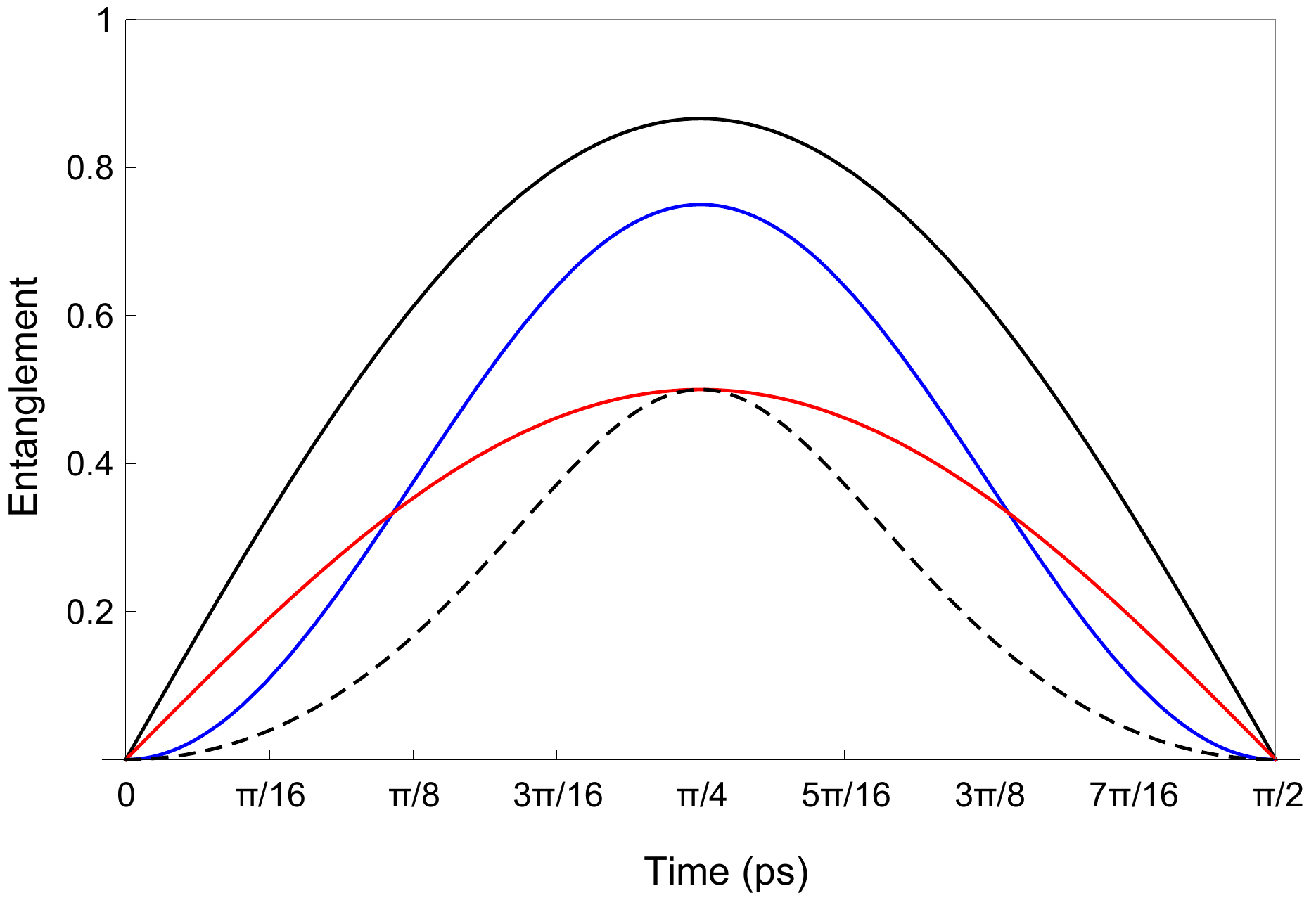}
\par\end{centering}

\caption{\label{fig:2}Dynamics of the entanglement measures for the simulation
with initial state $|\Psi(0)\rangle=|\uparrow\ocircle\rangle$. Concurrence
$\tilde{\mathcal{C}}\left(\Psi\right)$ is denoted with solid black
line, tangle $\tilde{\mathcal{T}}\left(\Psi\right)$ with solid blue
line, robustness of entanglement $\tilde{\mathcal{R}}\left(\Psi\right)$
with solid red line and the Schmidt number $\tilde{\mathcal{K}}\left(\Psi\right)$
with dashed black line. The coupling strength $\hbar\omega$ in the
interaction Hamiltonian between the two qutrits was set by $\omega=1$
rad/ps. Each cycle, between two consecutive separable states, lasts
$\frac{\pi}{2}$ ps.}
\end{figure}

Non-trivial quantum dynamics is obtained when the initial state is
not an energy eigenstate, but is a superposition of several energy
eigenstates with different eigenvalues. For $|\Psi(0)\rangle=|\uparrow\ocircle\rangle=\frac{1}{\sqrt{2}}\left(|E_{2}\rangle-|E_{6}\rangle\right)$
or $|\Psi(0)\rangle=|\ocircle\uparrow\rangle=\frac{1}{\sqrt{2}}\left(|E_{2}\rangle+|E_{6}\rangle\right)$,
the quantum state vector oscillates forth-and-back between the states
$|\uparrow\ocircle\rangle$ and $|\ocircle\uparrow\rangle$. Similarly,
for $|\Psi(0)\rangle=|\ocircle\downarrow\rangle=\frac{1}{\sqrt{2}}\left(|E_{4}\rangle-|E_{8}\rangle\right)$
or $|\Psi(0)\rangle=|\downarrow\ocircle\rangle=\frac{1}{\sqrt{2}}\left(|E_{4}\rangle+|E_{8}\rangle\right)$,
the quantum state vector oscillates forth-and-back between the states
$|\ocircle\downarrow\rangle$ and $|\downarrow\ocircle\rangle$. In
those four situations, the state vector $|\Psi(t)\rangle$ remains
confined within a two-dimensional subspace of the composite Hilbert
space, which implies that the entanglement measures cannot reach their
absolute maximum. The composite state vector has dynamic Schmidt coefficients
that depend on time
\begin{equation}
\begin{cases}
\lambda_{1} & =\left|\cos\left(\omega t\right)\right|\\
\lambda_{2} & =\left|\sin\left(\omega t\right)\right|\\
\lambda_{3} & =0 .
\end{cases}\label{eq:case-1}
\end{equation}
All four entanglement measures undergo cyclic dynamics (Fig. \ref{fig:1}).
Each cycle, between two consecutive separable states, lasts $\frac{\pi}{2}$
ps. Concurrence $\tilde{\mathcal{C}}\left(\Psi\right)$ is most sensitive
to the presence of quantum entanglement and forms an upper bound on
all four entanglement measures (Fig. \ref{fig:2}). Conversely, the
Schmidt number $\tilde{\mathcal{K}}\left(\Psi\right)$ is least sensitive
to the presence of quantum entanglement and forms a lower bound on
all four entanglement measures (Fig. \ref{fig:2}).

\begin{theorem}
The tangle $\tilde{\mathcal{T}}\left(\Psi\right)$ and robustness
of entanglement $\tilde{\mathcal{R}}\left(\Psi\right)$ are not partially
ordered.
\end{theorem}

\begin{proof}
The lack of partial order is demonstrated numerically (Fig.~\ref{fig:2}).
For example, consider the Schmidt coefficients \eqref{eq:case-1}
with the following assignments: $\tilde{\mathcal{R}}\left[\Psi\left(\omega t=\frac{\pi}{16}\right)\right]>\tilde{\mathcal{T}}\left[\Psi\left(\omega t=\frac{\pi}{16}\right)\right]$
and $\tilde{\mathcal{R}}\left[\Psi\left(\omega t=\frac{\pi}{4}\right)\right]<\tilde{\mathcal{T}}\left[\Psi\left(\omega t=\frac{\pi}{4}\right)\right]$.
\end{proof}

\subsection{Case 2}

\begin{figure}[t]
\begin{centering}
\includegraphics[width=140mm]{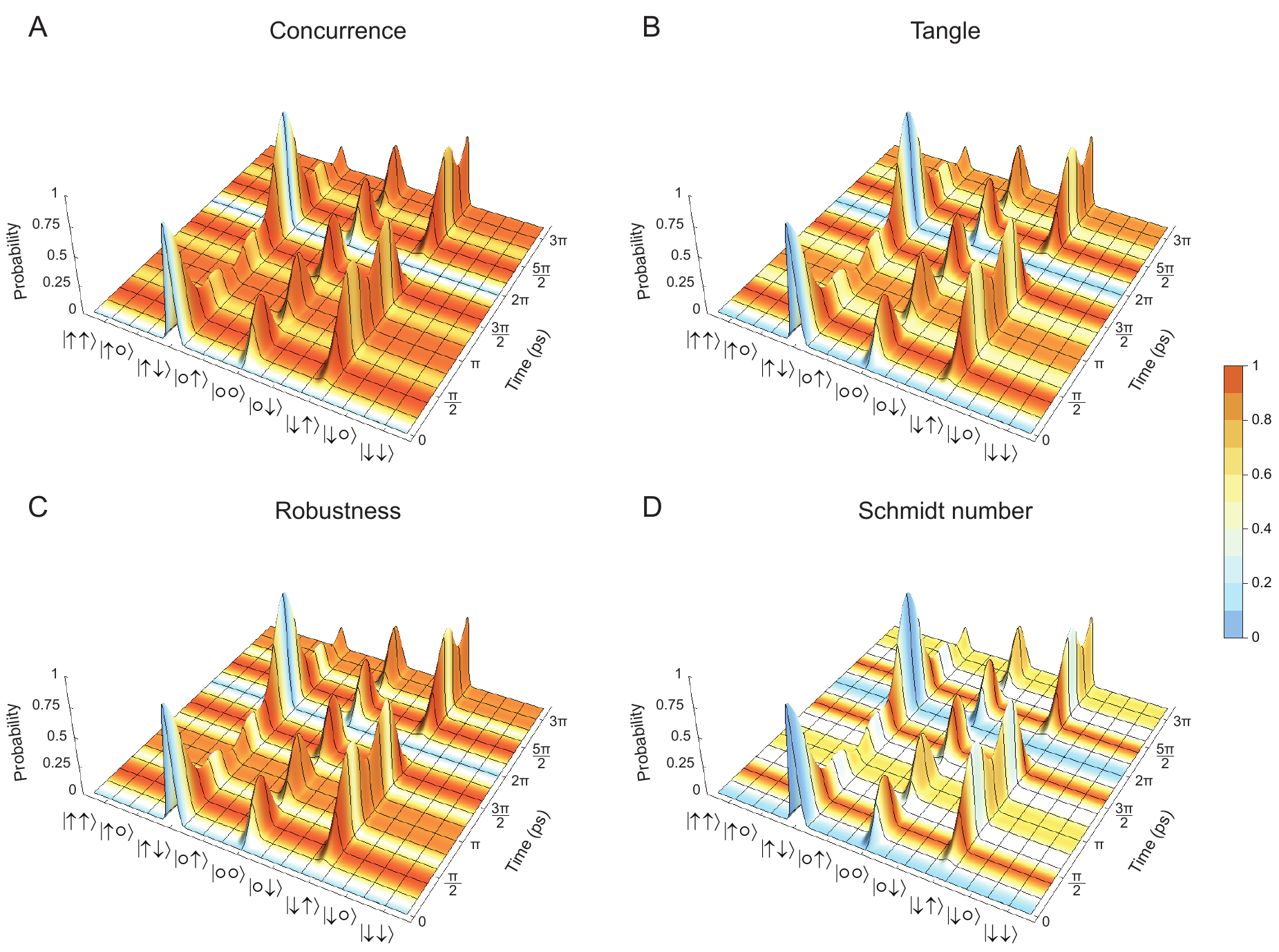}
\par\end{centering}

\caption{\label{fig:3}Dynamics of the expectation values for the respective
projectors onto the eigenstates of the quantum observable $\hat{\sigma}_{z}\otimes\hat{\sigma}_{z}$
simulated with the initial state $|\Psi(0)\rangle=|\uparrow\downarrow\rangle$.
The amount of quantum entanglement is measured with the use of concurrence
in panel (A), tangle in panel (B), robustness of entanglement in panel
(C) and the Schmidt number in panel (D). The coupling strength $\hbar\omega$
in the interaction Hamiltonian between the two qutrits was set by
$\omega=1$ rad/ps.}
\end{figure}

\begin{figure}[t]
\begin{centering}
\includegraphics[width=120mm]{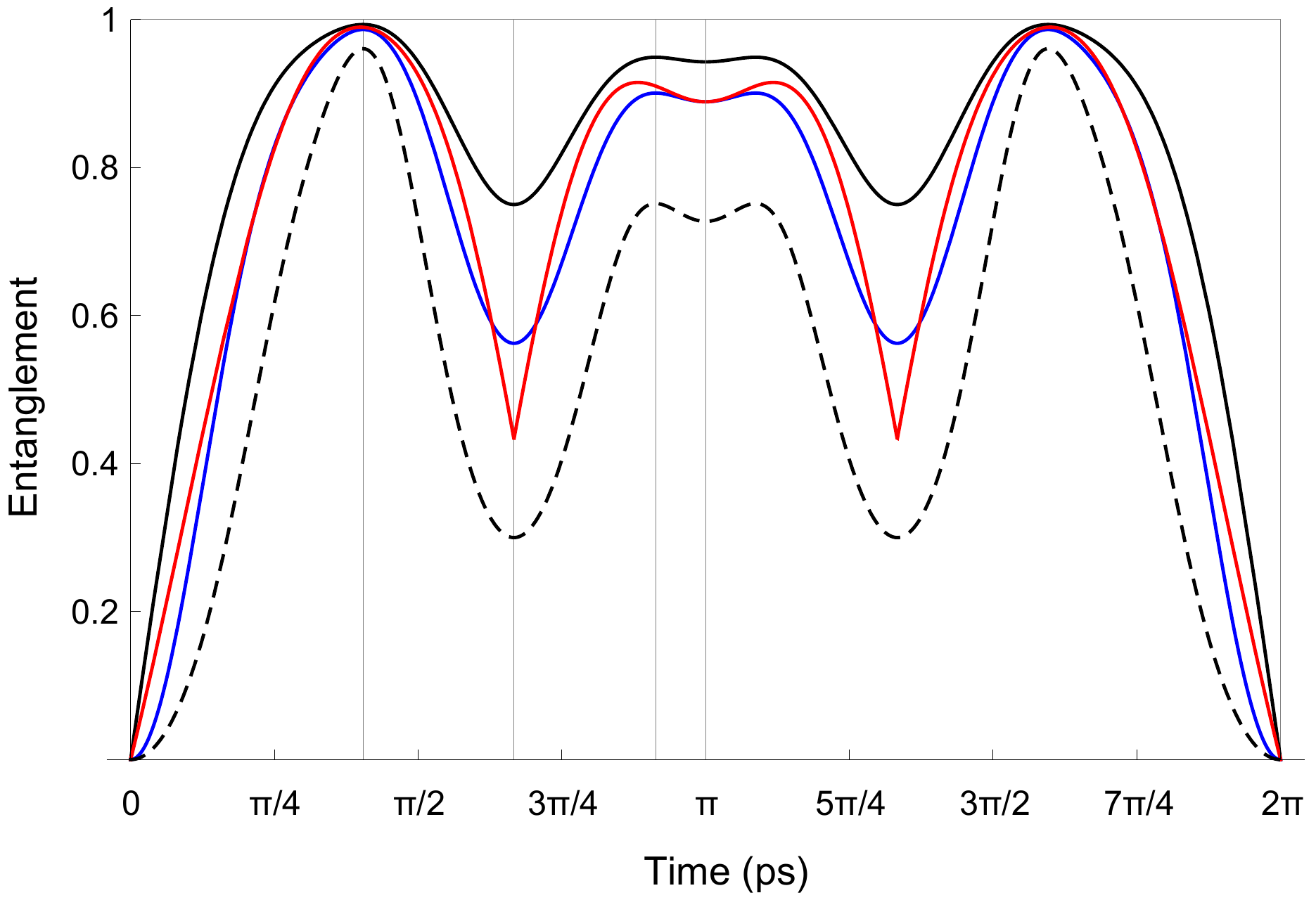}
\par\end{centering}

\caption{\label{fig:4}Dynamics of the entanglement measures for the simulation
with initial state $|\Psi(0)\rangle=|\uparrow\downarrow\rangle$.
Concurrence $\tilde{\mathcal{C}}\left(\Psi\right)$ is denoted with
solid black line, tangle $\tilde{\mathcal{T}}\left(\Psi\right)$ with
solid blue line, robustness of entanglement $\tilde{\mathcal{R}}\left(\Psi\right)$
with solid red line and the Schmidt number $\tilde{\mathcal{K}}\left(\Psi\right)$
with dashed black line. The coupling strength $\hbar\omega$ in the
interaction Hamiltonian between the two qutrits was set by $\omega=1$
rad/ps. Each cycle, between two consecutive separable states, lasts
$2\pi$ ps.}
\end{figure}

Complicated quantum dynamics, manifesting varying quantum interference
effects, occurs when $|\Psi(0)\rangle=|\uparrow\downarrow\rangle=\frac{1}{\sqrt{6}}|E_{3}\rangle-\frac{1}{\sqrt{2}}|E_{7}\rangle+\frac{1}{\sqrt{3}}|E_{9}\rangle$
or $|\Psi(0)\rangle=|\downarrow\uparrow\rangle=\frac{1}{\sqrt{6}}|E_{3}\rangle+\frac{1}{\sqrt{2}}|E_{7}\rangle+\frac{1}{\sqrt{3}}|E_{9}\rangle$.
In those situations, the quantum state vector $|\Psi(t)\rangle$ oscillates
with somewhat irregular pattern between the states $|\uparrow\downarrow\rangle$,
$|\ocircle\ocircle\rangle$ and $|\downarrow\uparrow\rangle$. The
composite state vector has dynamic Schmidt coefficients
\begin{equation}
\begin{cases}
\lambda_{1} & =\frac{1}{3}\sqrt{\frac{7}{2}+3\cos\left(\omega t\right)+\frac{3}{2}\cos\left(2\omega t\right)+\cos\left(3\omega t\right)}\\
\lambda_{2} & =\frac{2}{3}\sqrt{5+4\cos\left(\omega t\right)}\sin^{2}\left(\frac{1}{2}\omega t\right)\\
\lambda_{3} & =\frac{2}{3}\left|\sin\left(\frac{3}{2}\omega t\right)\right| .
\end{cases}
\end{equation}
All four entanglement measures undergo cyclic dynamics (Fig. \ref{fig:3}).
Each cycle, between two consecutive separable states, lasts $2\pi$
ps. Again, concurrence $\tilde{\mathcal{C}}\left(\Psi\right)$ and
the Schmidt number $\tilde{\mathcal{K}}\left(\Psi\right)$ form bounds
on the four entanglement measures from above and below, respectively
(Fig. \ref{fig:4}). Because, in this simulation it is possible for
all 3 Schmidt coefficients to be non-zero, the entanglement measures
explore almost the full range from separable to maximally entangled
state (Fig. \ref{fig:4}). Again, it is clearly seen in the vicinity
of local minima that the tangle $\tilde{\mathcal{T}}\left(\Psi\right)$
and robustness of entanglement $\tilde{\mathcal{R}}\left(\Psi\right)$
and are not partially ordered (Fig. \ref{fig:4}).

\subsection{Case 3}

\begin{figure}[t]
\begin{centering}
\includegraphics[width=140mm]{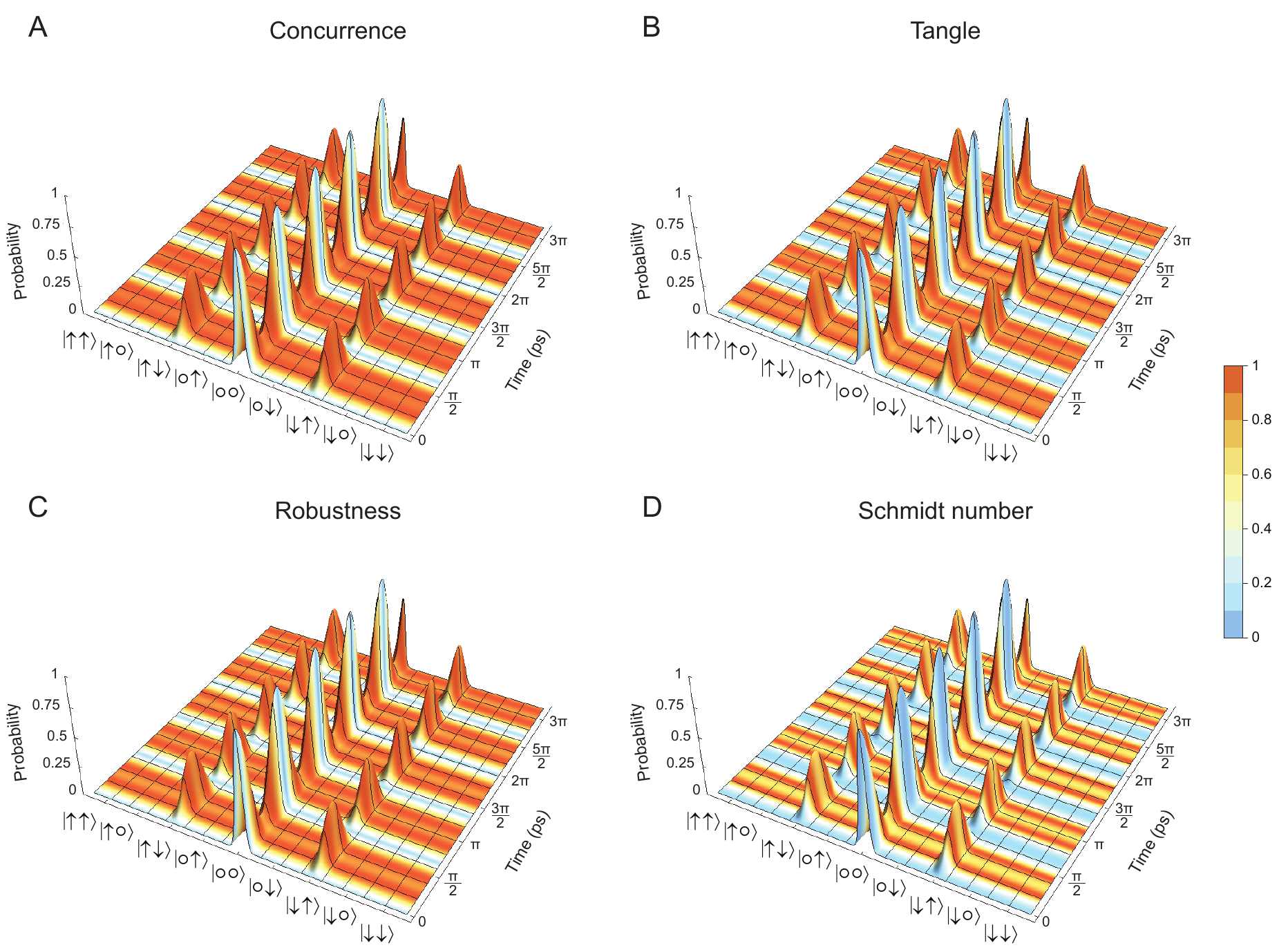}
\par\end{centering}

\caption{\label{fig:5}Dynamics of the expectation values for the respective
projectors onto the eigenstates of the quantum observable $\hat{\sigma}_{z}\otimes\hat{\sigma}_{z}$
simulated with the initial state $|\Psi(0)\rangle=|\ocircle\ocircle\rangle$.
The amount of quantum entanglement is measured with the use of concurrence
in panel (A), tangle in panel (B), robustness of entanglement in panel
(C) and the Schmidt number in panel (D). The coupling strength $\hbar\omega$
in the interaction Hamiltonian between the two qutrits was set by
$\omega=1$ rad/ps.}
\end{figure}

\begin{figure}[t]
\begin{centering}
\includegraphics[width=120mm]{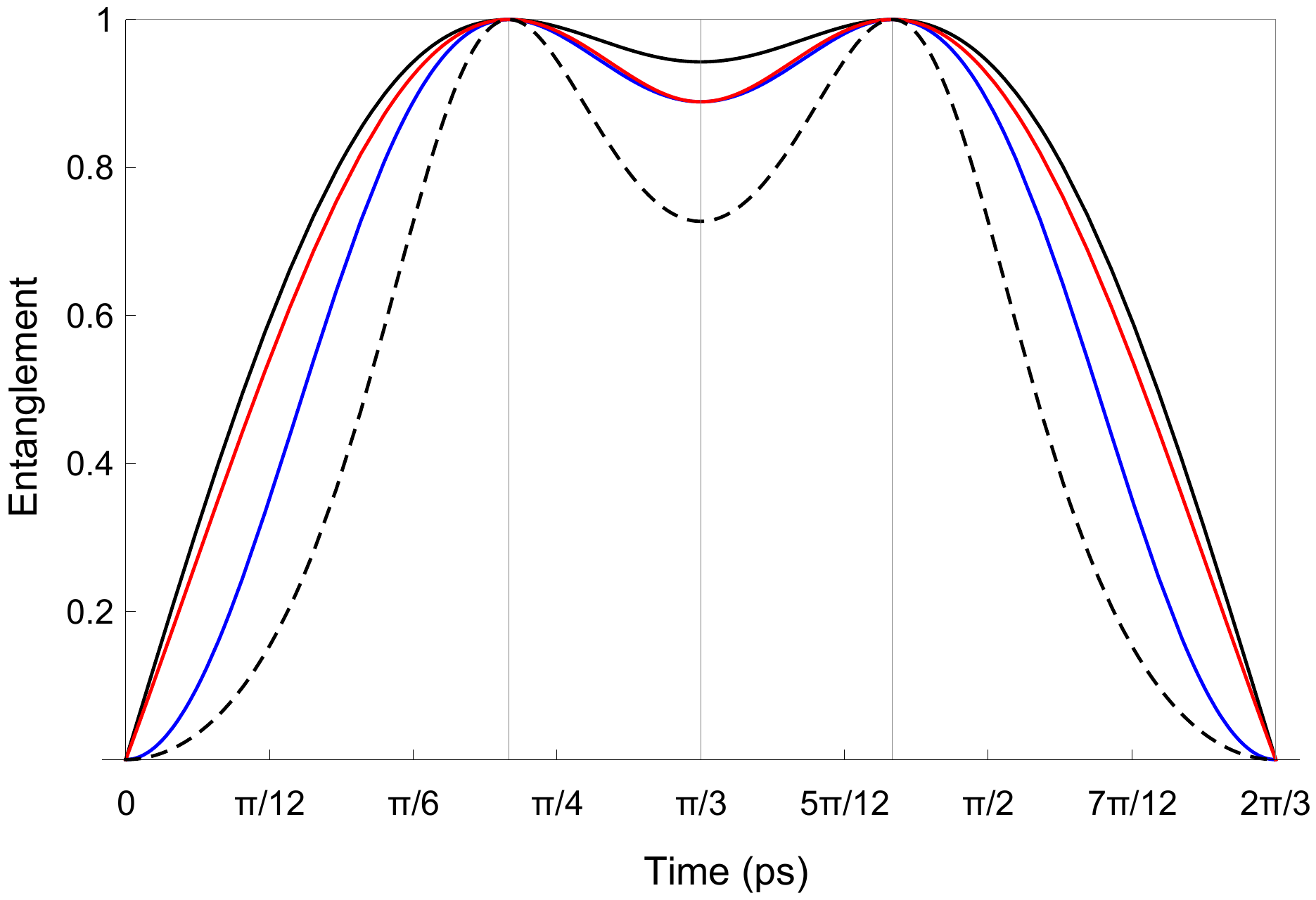}
\par\end{centering}

\caption{\label{fig:6}Dynamics of the entanglement measures for the simulation
with initial state $|\Psi(0)\rangle=|\ocircle\ocircle\rangle$. Concurrence
$\tilde{\mathcal{C}}\left(\Psi\right)$ is denoted with solid black
line, tangle $\tilde{\mathcal{T}}\left(\Psi\right)$ with solid blue
line, robustness of entanglement $\tilde{\mathcal{R}}\left(\Psi\right)$
with solid red line and the Schmidt number $\tilde{\mathcal{K}}\left(\Psi\right)$
with dashed black line. The coupling strength $\hbar\omega$ in the
interaction Hamiltonian between the two qutrits was set by $\omega=1$
rad/ps. Each cycle, between two consecutive separable states, lasts
$\frac{2\pi}{3}$ ps.}
\end{figure}

Quantum dynamics with regular quantum interference pattern occurs
when $|\Psi(0)\rangle=|\ocircle\ocircle\rangle=\sqrt{\frac{2}{3}}|E_{3}\rangle-\frac{1}{\sqrt{3}}|E_{9}\rangle$.
In this case, the state vector $|\Psi(t)\rangle$ explores the complete
spectrum from separable to maximally entangled states. The composite
state vector has dynamic Schmidt coefficients
\begin{equation}
\begin{cases}
\lambda_{1} & =\frac{1}{3}\sqrt{5+4\cos\left(3\omega t\right)}\\
\lambda_{2} & =\frac{2}{3}\left|\sin\left(\frac{3}{2}\omega t\right)\right|\\
\lambda_{3} & =\frac{2}{3}\left|\sin\left(\frac{3}{2}\omega t\right)\right| .
\end{cases}
\end{equation}
All four entanglement measures undergo cyclic dynamics (Fig. \ref{fig:5}).
Each cycle, between two consecutive separable states, lasts $\frac{2\pi}{3}$
ps. The maximal entanglement reaches 1 at $\omega t=\frac{2\pi}{9}$
and $\omega t=\frac{4\pi}{9}$ (Fig.~\ref{fig:6}). Due to its high
sensitivity to quantum entanglement, the concurrence $\tilde{\mathcal{C}}\left(\Psi\right)$
is able to resolve quite sharply the separable states (narrow blue
bands in Fig.~\ref{fig:5}A). Conversely, due to its low sensitivity
to quantum entanglement, the Schmidt number $\tilde{\mathcal{K}}\left(\Psi\right)$
is able to resolve quite sharply the maximally entangled states (narrow
red bands in Fig.~\ref{fig:5}D).

The sharpness of the bands produced in contour plots depends on the temporal rate of change of the entanglement measures.
By construction, all normalized entanglement measures coincide at the two extreme cases: at separable states, where they have zero value, and at maximally entangled states, where they have unit value \cite{Vedral1997,Vedral1998,Vidal2000,Donald2002}. The existence of partial order, $\mathcal{A}(t)\geq \mathcal{B}(t)$, then implies that in the neighborhood of a separable state, the quantum dynamics is described by U- or V-shaped cups such that the cup for $\mathcal{A}(t)$ is inside the cup for $\mathcal{B}(t)$ (Fig.~\ref{fig:6}). Consequently, the time derivative for $\mathcal{A}(t)$ decreases towards and raises from the separable state more steeply compared to the time derivative for $\mathcal{B}(t)$, namely, $\left|\frac{d\mathcal{A}(t)}{dt}\right|\geq\left|\frac{d\mathcal{B}(t)}{dt}\right|$. The~roles played by the two entanglement measures are reversed in the neighborhood of a maximally entangled state, where the quantum dynamics is described by U- or V-shaped caps such that the cap for $\mathcal{B}(t)$ is inside the cap for $\mathcal{A}(t)$ (Fig.~\ref{fig:6}). This means that the time derivative for $\mathcal{B}(t)$ raises towards and decreases from the maximally entangled state more steeply compared to the time derivative for $\mathcal{A}(t)$, namely, $\left|\frac{d\mathcal{B}(t)}{dt}\right|\geq\left|\frac{d\mathcal{A}(t)}{dt}\right|$.
Different time derivatives result in different durations of the time windows for which the entanglement measures stay inside some fixed small region $\varepsilon >0$ near 0 or 1. In general, steeper time derivatives imply shorter time windows and sharper resolution of quantum dynamics.

\section{Entanglement statistics}
\label{sec:6}

To further illustrate how the Schmidt decomposition features in the calculation of measurable correlations between the outcomes of quantum observables, next we introduce the concept of state-dependent \emph{entanglement statistics}. We also derive an alternative, but equivalent, characterization of separable/entangled states in terms of absence/presence of correlations between all/some local quantum observables.

We begin with the study of quantum operator statistics. Let $\lscript _S(H)$ be the set of self-adjoint (Hermitian) operators \cite{Sakurai2020,Walecka2021} on a finite-dimensional complex Hilbert space~$H$.
We call the elements of $\lscript _S(H)$ \textit{observable operators}. If $\ket{\phi}\in H$ is a vector state and $\hat{A}\in\lscript _S(H)$ is an operator, the \textit{expectation} (or \textit{average}) of $\hat{A}$ in the state $\ket{\phi}$ is
\begin{equation}
\elbows{\hat{A}}_\phi =\elbows{\phi | \hat{A} |\phi}. 
\end{equation}
The \textit{deviation} of $\hat{A}$ in the state $\ket{\phi}$ is
\begin{equation}
\hat{D}_\phi (\hat{A})=\hat{A}-\elbows{\hat{A}}_\phi \hat{I}
\end{equation}
where $\hat{I}\in\lscript _S(H)$ is the identity operator. An interaction between any two observable operators $\hat{A},\hat{B}\in\lscript _S(H)$ is described by the $\phi$-\textit{correlation} given by
\begin{equation}
\rmcor _\phi (\hat{A},\hat{B})=\elbows{\phi |\hat{D}_\phi (\hat{A})\hat{D}_\phi (\hat{B})|\phi} .
\end{equation}
In general, $\rmcor _\phi (\hat{A},\hat{B})$ is not a real number and we have
\begin{equation}
\overline{\rmcor _\phi (\hat{A},\hat{B})}=\rmcor _\phi (\hat{B},\hat{A}) .
\end{equation}
We define the $\phi$-\textit{covariance} of $\hat{A}$ and $\hat{B}$ as
\begin{equation}
\Delta _\phi (\hat{A},\hat{B}) =\rmre \left[\rmcor _\phi (\hat{A},\hat{B})\right]
\end{equation}
and the $\phi$-\textit{variance} of $\hat{A}$ as
\begin{equation}
\Delta _\phi (\hat{A}) =\Delta _\phi (\hat{A},\hat{A})=\rmcor _\phi (\hat{A},\hat{A}) .
\end{equation}
We have that 
\begin{align}
\rmcor _\phi (\hat{A},\hat{B})&=\elbows{\phi |(\hat{A}-\elbows{\hat{A}}_\phi \hat{I})(\hat{B}-\elbows{\hat{B}}_\phi \hat{I})|\phi}
   =\elbows{\hat{A}\hat{B}}_\phi-\elbows{\hat{A}}_\phi\elbows{\hat{B}}_\phi \label{eq:corr},\\
\Delta _\phi (\hat{A},\hat{B})&=\rmre \left[\elbows{\hat{A}\hat{B}}_\phi-\elbows{\hat{A}}_\phi\elbows{\hat{B}}_\phi \right] ,\\
\Delta _\phi (\hat{A})&=\elbows{\hat{A}^2}_\phi-\elbows{\hat{A}}_\phi ^2 .
\end{align}
If $\rmcor _\phi (\hat{A},\hat{B})=0$, we say that $\hat{A}$ and $\hat{B}$ are $\phi$-\textit{uncorrelated}. Of course, from \eqref{eq:corr} it follows that $\rmcor _\phi (\hat{A},\hat{B})=0$ if and only if
$\elbows{\phi |\hat{A}\hat{B}|\phi}=\elbows{\hat{A}}_\phi\elbows{\hat{B}}_\phi$.
Since quantum observables are represented by self-adjoint (Hermitian) operators \cite{Sakurai2020,Walecka2021}
\begin{equation}
\overline{\elbows{\phi |\hat{A}\hat{B}|\phi}}=\elbows{\phi|(\hat{A}\hat{B})^\dagger|\phi}=\elbows{\phi |\hat{B}^\dagger \hat{A}^\dagger |\phi}=\elbows{\phi |\hat{B} \hat{A} |\phi} ,
\end{equation}
it follows that if $\hat{A}$ and $\hat{B}$ are $\phi$-uncorrelated then $\hat{B}$ and $\hat{A}$ are $\phi$-uncorrelated. For $\hat{A},\hat{B}\in\lscript _S(H)$, their \textit{commutator} is given by
$\sqbrac{\hat{A},\hat{B}}=\hat{A}\hat{B}-\hat{B}\hat{A}$. An important connection between these concepts is the following \textit{uncertainty principle} \cite{Gudder2022}:
\begin{equation}                
\label{eq1}
\frac{1}{4}\ab{\elbows{\phi |\sqbrac{\hat{A},\hat{B}}|\phi}}^2+\sqbrac{\Delta _\phi (\hat{A},\hat{B})}^2=\ab{\rmcor _\phi (\hat{A},\hat{B})}^2\le\Delta _\phi (\hat{A})\Delta _\phi (\hat{B}) .
\end{equation}
As a special case, we have the Heisenberg--Robertson uncertainty principle \cite{Robertson1929,Baggott2020}
\begin{equation}
\frac{1}{4}\ab{\elbows{\phi |\sqbrac{\hat{A},\hat{B}}|\phi}}^2\le\Delta _\phi (\hat{A})\Delta _\phi (\hat{B}) .
\end{equation}
It follows from \eqref{eq1} that $\hat{A}$ and $\hat{B}$ are uncorrelated if and only if
\begin{equation}
\elbows{\phi |\sqbrac{\hat{A},\hat{B}} |\phi}=\Delta _\phi (\hat{A},\hat{B})=0 .
\end{equation}
In this case, \eqref{eq1} gives no information.

An operator $\hat{A}\in\lscript _S(H)$ satisfying $0\le \hat{A}\le \hat{I}$ is called an \textit{effect} \cite{Heinosaari2012}. Effects correspond to quantum events or yes-no measurements. A \textit{real-valued observable} is a set of effects $\ascript =\brac{\hat{A}_x\colon x\in\Omega _\ascript}$ where $\Omega _\ascript$ is a finite subset of $\real$ called the \textit{outcome set} of $\ascript$ and we have $\sum _{x\in\Omega _\ascript}\hat{A}_x=\hat{I}$. We consider $\hat{A}_x$ to be the event that occurs when a measurement of $\ascript$ results in the outcome $x$. The \textit{probability} that $\ascript$ has the outcome $x$ when the system is in the state $\ket{\phi}\in H$ is $\elbows{\hat{A}_x}_\phi =\elbows{\phi |\hat{A}_x|\phi}$. Notice that $x\mapsto\elbows{\phi |\hat{A}_x|\phi}$ is a probability measure because
\begin{equation}
\sum _{x\in\Omega _\ascript}\elbows{\phi |\hat{A}_x|\phi}=\elbows{\phi |\sum_{x\in\Omega _\ascript}\hat{A}_x|\phi}=\elbows{\phi |\hat{I}|\phi}=\elbows{\phi |\phi}=1 .
\end{equation}
Corresponding to the observable $\ascript$ we have its \textit{stochastic operator} $\atilde =\sum _{x\in\Omega _\ascript}x\hat{A}_x$. Then $\atilde$ is an observable operator and the \textit{expectation} (or \textit{average}) of $\ascript$ in the state $\ket{\phi}$ is
\begin{equation}
\elbows{\ascript}_\phi =\elbows{\atilde}_\phi =\elbows{\phi |\atilde|\phi}=\elbows{\phi |\sum _{x\in\Omega _\ascript}x\hat{A}_x|\phi}
   =\sum _{x\in\Omega _\ascript}x\elbows{\hat{A}_x}_\phi .
\end{equation}
Thus, $\elbows{\ascript}_\phi$ is the sum of the outcomes of $\ascript$ times the probabilities that these outcomes occur. If
$\bscript=\brac{\hat{B}_y\colon y\in\Omega _\bscript}$ is another observable, the $\phi$-\textit{correlation} of $\ascript$ and $\bscript$ is
$\rmcor _\phi (\ascript ,\bscript )=\rmcor _\phi (\atilde ,\btilde )$, the $\phi$-\textit{covariance} of $\ascript$ and $\bscript$ is
$\Delta _\phi (\ascript ,\bscript )=\Delta _\phi (\atilde ,\btilde )$ and the $\phi$-\text{variance} of $\ascript$ is
$\Delta _\phi (\ascript )=\Delta _\phi (\atilde\,)$. We conclude that
\begin{align}                
\rmcor _\phi (\ascript ,\bscript )&=\sum _{x,y}xy\paren{\elbows{\phi |\hat{A}_x\hat{B}_y|\phi}-\elbows{\hat{A}_x}_\phi\elbows{\hat{B}_y}_\phi}\label{eq2} ,\\
\Delta _\phi (\ascript ,\bscript )&=\sum _{xy}xy \,\rmre\left[\elbows{\phi |\hat{A}_x\hat{B}_y|\phi}-\elbows{\hat{A}_x}_\phi\elbows{\hat{B}_y}_\phi\right]\label{eq3} ,\\
\Delta _\phi (\ascript )&=\sum _{x,y}xy\paren{\elbows{\phi |\hat{A}_x\hat{A}_y|\phi}-\elbows{\hat{A}_x}_\phi\elbows{\hat{A}_y}_\phi}\label{eq4} .
\end{align}
We say that $\ascript ,\bscript$ are $\phi$-\textit{uncorrelated} if $\atilde$, $\btilde$ are $\phi$-uncorrelated. Also, $\ascript$, $\bscript$ are $\phi$-\textit{independent} if
\begin{equation}
\elbows{\phi |\hat{A}_x\hat{B}_y|\phi}=\elbows{\hat{A}_x}_\phi\elbows{\hat{B}_y}_\phi
\end{equation}
for all $x\in\Omega _\ascript$, $y\in\Omega _\bscript$. Of course, if $\ascript$, $\bscript$ are $\phi$-independent, then $\bscript$, $\ascript$ are $\phi$-independent. It follows from
\eqref{eq2} that if $\ascript$, $\bscript$ are $\phi$-independent, then $\ascript$, $\bscript$ are $\phi$-uncorrelated. The converse does not hold because there are examples of uncorrelated random variables that are not independent. Two observables $\ascript$, $\bscript$
\textit{commute} if $\sqbrac{\hat{A}_x,\hat{B}_y}=0$ for all $x\in\Omega _\ascript$, $y\in\Omega _\bscript$. If $\ascript$, $\bscript$ commute then so do $\atilde$, $\btilde$ and the uncertainty principle reduces to
\begin{equation}
\sqbrac{\Delta _\phi (\ascript ,\bscript )}^2=\ab{\rmcor _\phi (\ascript ,\bscript )}^2\le\Delta _\phi (\ascript )\Delta _\phi (\bscript ) .
\end{equation}
We say that $\ascript$, $\bscript$ are \textit{compatible} (or \textit{jointly measurable}) if there exists an observable
$\cscript =\brac{\hat{C}_{xy}\colon x\in\Omega _\ascript ,y\in\Omega _\bscript}$ such that $\hat{A}_x=\sum _{y\in\Omega _\bscript}\hat{C}_{xy}, \hat{B}_y=\sum _{x\in\Omega _\ascript}\hat{C}_{xy}$.
We then call $\cscript$ a \textit{joint observable} for $\ascript$, $\bscript$.
Although $\cscript$ is an observable, it is not real-valued because its outcome space is $\Omega _\ascript\times\Omega _\bscript$.
If $\ascript$, $\bscript$ commute, they are compatible with joint observable $\hat{C}_{xy}=\hat{A}_x\hat{B}_y$. If $\ascript$, $\bscript$ are compatible, they need not commute \cite{Gudder2022}.

We now apply the previous discussion to \emph{entanglement statistics}. Let $H_1$, $H_2$ be finite-dimensional Hilbert spaces for two quantum systems. The combined system is described by the Hilbert space $H_1\otimes H_2$. If $\hat{A}\in\lscript _S(H_1)$ is an observable operator for system~1, then in the combined system this operator is represented by $\hat{A}\otimes \hat{I}_2$ where $\hat{I}_2$ is the identity operator on $H_2$. Similarly, if $\hat{B}\in\lscript _S(H_2)$, then in the combined system $\hat{B}$ is represented by $\hat{I}_1\otimes \hat{B}$. Since $\hat{A}\otimes \hat{I}_2$ and $\hat{I}_1\otimes \hat{B}$ commute, they are jointly measured by the observable
\begin{equation}
(\hat{A}\otimes \hat{I}_2)(\hat{I}_1\otimes \hat{B})=\hat{A}\otimes \hat{B}
\end{equation}
in the combined system. Similarly, if $\ascript =\brac{\hat{A}_x\colon x\in\Omega _\ascript}$, $\bscript =\brac{\hat{B}_y\colon y\in\Omega _\bscript}$ are real-valued observables on
$H_1$, $H_2$, respectively, then
$$
\ascript\otimes \hat{I}_2=\brac{\hat{A}_x\otimes \hat{I}_2\colon x\in\Omega _\ascript}, \hat{I}_1\otimes\bscript =\brac{\hat{I}_1\otimes \hat{B}_y\colon y\in\Omega _\bscript}
$$
are real-valued observables on $H_1\otimes H_2$.
These observables commute and have joint observable
\begin{equation}
\cscript =\ascript\otimes\bscript =\brac{\hat{C}_{xy}=\hat{A}_x\otimes \hat{B}_y\colon (x,y)\in\Omega _\ascript\times\Omega _\bscript} .
\end{equation}
Although $\cscript$ is an observable on $H_1\otimes H_2$, it is not real-valued because $\Omega _\cscript =\Omega _\ascript\times\Omega _\bscript$.
The corresponding stochastic operators become
\begin{align}
\widetilde{(\ascript\otimes \hat{I}_2)} &=\atilde\otimes \hat{I}_2=\sum _{x\in\Omega _\ascript}x\hat{A}_x\otimes \hat{I}_2 ,\\
\widetilde{(\hat{I}_1\otimes\bscript )} &= \hat{I}_1\otimes\btilde =\sum _{y\in\Omega _\bscript}y\hat{I}_1\otimes  \hat{B}_y .
\end{align}
We define the \textit{stochastic operator} for $\cscript$ to be
\begin{equation}
\ctilde =\widetilde{(\ascript\otimes \hat{I}_2)} \widetilde{(\hat{I}_1\otimes\bscript )} =\atilde\otimes\btilde =\sum _{x,y}xy\hat{A}_x\otimes \hat{B}_y .
\end{equation}

\begin{theorem}    
\label{thm1}
The following statements are equivalent:\\
{\rm{(i)}}\enspace The vector state $\ket{\alpha}\in H_1\otimes H_2$ is separable.\\
{\rm{(ii)}}\enspace $\hat{A}\otimes \hat{I}_2$, $\hat{I}_1\otimes \hat{B}$ are $\alpha$-uncorrelated for all $\hat{A}\in\lscript _S(H_1)$, $\hat{B}\in\lscript _S(H_2)$.\\
{\rm{(iii)}}\enspace $\ascript\otimes \hat{I}_2$, $\hat{I}_1\otimes\bscript$ are $\alpha$-independent for all observables $\ascript$ on $H_1$ and $\bscript$ on $H_2$.
\end{theorem}
\begin{proof}
(i)$\Rightarrow$(ii)\enspace Suppose $\ket{\alpha}\in H_1\otimes H_2$ is separable with $\ket{\alpha} =\ket{\phi}\otimes\ket{\psi}$. We then have
\begin{align*}
\elbows{\alpha |(\hat{A}\otimes \hat{I}_2)(\hat{I}_1\otimes \hat{B})|\alpha}&=\elbows{\phi|\otimes\langle\psi |\hat{A}\otimes \hat{B}|\phi\rangle\otimes|\psi}
=\langle\phi|\otimes\langle \psi |\left(\hat{A}|\phi\rangle\otimes \hat{B}|\psi\rangle\right)\\
&=\elbows{\phi |\hat{A}|\phi}\elbows{\psi |\hat{B}|\psi}=\elbows{\alpha |\hat{A}\otimes \hat{I}_2|\alpha}\elbows{\alpha |\hat{I}_1\otimes \hat{B}|\alpha} .
\end{align*}
It follows that $\hat{A}\otimes \hat{I}_2$, $\hat{I}_1\otimes \hat{B}$ are $\alpha$-uncorrelated.

(ii)$\Rightarrow$(iii)\enspace Suppose $\hat{A}\otimes \hat{I}_2$, $\hat{I}_1\otimes \hat{B}$ are $\alpha$-uncorrelated for all $\hat{A}\in\lscript _S(H)$, $\hat{B}\in\lscript _S(H)$. Since
$\hat{A}_x\in\lscript _S(H)$, $\hat{B}_y\in\lscript _S(H)$ we have that $\hat{A}_x\otimes \hat{I}_2$, $\hat{I}_1\otimes \hat{B}_y$ are $\alpha$-uncorrelated for all $x\in\Omega _\ascript$,
$y\in\Omega _\bscript$. Hence,
$$
\elbows{\alpha |(\hat{A}_x\otimes \hat{I}_2)(\hat{I}_1\otimes \hat{B}_y)|\alpha}=\elbows{\hat{A}\otimes \hat{I}_2}_\alpha\elbows{\hat{I}_1\otimes \hat{B}}_\alpha
$$
so the observables $\ascript\otimes \hat{I}_2$, $\hat{I}_1\otimes\bscript$ are $\alpha$-independent.

(iii)$\Rightarrow$(i)\enspace Suppose $\ascript\otimes \hat{I}_2$, $\hat{I}_1\otimes\bscript$ are $\alpha$-independent for all observables $\ascript$ on $H_1$ and $\bscript$ on
$H_2$. Let $\ket{\alpha}$ have Schmidt decomposition $\ket{\alpha} =\sum\lambda _i\ket{\phi}_i\otimes\ket{\psi}_i$. Since
\begin{equation}
\elbows{\alpha | \hat{A}_x\otimes \hat{B}_y |\alpha}=\elbows{\hat{A}_x\otimes \hat{I}_2}_\alpha\elbows{\hat{I}_1\otimes \hat{B}_y}_\alpha
\end{equation}
for all $x\in\Omega _\ascript$, $y\in\Omega _\beta$ and
\begin{align}            
\label{eq5}
\elbows{\hat{A}_x\otimes \hat{I}_2}_\alpha&=\elbows{\alpha |\hat{A}_x\otimes \hat{I}_2|\alpha}
   =\sum _i\lambda _i\langle\phi _i|\otimes\langle\psi_i|(\hat{A}_x\otimes \hat{I}_2)\sum _j\lambda _j|\phi _j\rangle\otimes|\psi _j\rangle\nonumber\\
   &=\sum _{i,j}\lambda _i\lambda _j\elbows{\phi _i|\otimes\langle\psi _i|(\hat{A}_x\otimes \hat{I}_2)|\phi _j\rangle\otimes|\psi _j}
   =\sum _{i,j}\lambda _i\lambda _j\elbows{\phi _i|\hat{A}_x|\phi _j}\elbows{\psi _i|\psi _j}\nonumber\\
   &=\sum _{i,j}\lambda _i\lambda _j\elbows{\phi _i|\hat{A}_x|\phi _j}\delta_{ij}=\sum _i\lambda _i^2\elbows{\phi _i|\hat{A}_x|\phi _i}
\end{align}
we conclude that
\begin{align}            
\label{eq6}
\sum _i\lambda _i^2\elbows{\phi _i|\hat{A}_x|\phi _i}&\sum _j\lambda _j^2\elbows{\psi _j|\hat{B}_y|\psi _j}\notag\\
   &=\elbows{\alpha |\hat{A}_x\otimes \hat{B}_y|\alpha}
   =\sum _i\lambda _i\langle\phi _i|\otimes\langle\psi _i|(\hat{A}_x\otimes \hat{B}_y ) \sum _j\lambda _j|\phi _j\rangle\otimes|\psi _j\rangle\nonumber\\
   &=\sum _{i,j}\lambda _i\lambda _j\langle\phi _i|\otimes\langle\psi _i|\left(\hat{A}_x|\phi _j\rangle\otimes \hat{B}_y|\psi _j\rangle\right)\nonumber\\
   &=\sum _{i,j}\lambda _i\lambda _j\elbows{\phi _i|\hat{A}_x|\phi _j}\elbows{\psi _i|\hat{B}_y|\psi _j} .
\end{align}
Since \eqref{eq6} holds for all $\ascript$, $\bscript$, we can let $\hat{A}_x=\ket{\phi _1}\bra{\phi_1}$, $\hat{B}_y=\ket{\psi _1}\bra{\psi _1}$ to obtain $\lambda _1^4=\lambda _1^2$.
Although $\lambda _1$ could be either 0 or 1 to satisfy the latter relation, we can also use the assumption that the Schmidt coefficients are sorted in descending order to fix $\lambda _1=1$.
Therefore, $\ket{\alpha} =\ket{\phi _1}\otimes\ket{\psi _1}$ is separable.
\end{proof}

If $\ket{\alpha}$ is entangled with Schmidt decomposition $\ket{\alpha} =\sum\lambda _i\ket{\phi _i}\otimes\ket{\psi _i}$, we have seen in \eqref{eq6} that for all $\hat{A}\in\lscript _S(H)$,
$\hat{B}\in\lscript _S(H)$ we have
\begin{equation}                
\label{eq7}
\elbows{\hat{A}\otimes \hat{B}}_\alpha =\sum _{i,j}\lambda _i\lambda _j\elbows{\phi _i|\hat{A}|\phi _j}\elbows{\psi _i|\hat{B}|\psi _j} .
\end{equation}
If $\ket{\alpha} =\ket{\phi}\otimes\ket{\psi}$ is separable, then
\begin{equation}
\elbows{\hat{A}\otimes \hat{B}}_\alpha =\elbows{\hat{A}}_\phi\elbows{\hat{B}}_\psi =\elbows{\hat{A}\otimes \hat{I}_2}_\alpha\elbows{\hat{I}_1\otimes \hat{B}}_\alpha .
\label{eq:94}
\end{equation}
At the other extreme, if the normalized entanglement number is $\tilde{e}(\alpha )=1$, then $\ket{\alpha}$~is~maximally entangled with Schmidt decomposition
$\ket{\alpha} =\tfrac{1}{\sqrt{n\,}}\sum\limits _{i=1}^n\ket{\phi _i}\otimes\ket{\psi _i}$ where
$n=\min\paren{\dim (H_1),\dim (H_2)}$.
In this case, we have
\begin{equation}
\elbows{\hat{A}\otimes \hat{B}}_\alpha =\frac{1}{n}\sum _{i,j}\elbows{\phi _i|\hat{A}|\phi _j}\elbows{\psi _i|\hat{B}|\psi _j} .
\end{equation}

For a general state $\ket{\alpha}\in H_1\otimes H_2$ the $\alpha$-variance of $\hat{A}\otimes \hat{B}$ becomes
\begin{equation}
\Delta _\alpha (\hat{A}\otimes \hat{B})=\elbows{\alpha |\hat{A}^2\otimes \hat{B}^2|\alpha}-\elbows{\hat{A}\otimes \hat{B}}_\alpha ^2 .
\end{equation}
If $\ket{\alpha} =\ket{\phi}\otimes\ket{\psi}$ is separable, we obtain
\begin{align}
\Delta _\alpha (\hat{A}\otimes \hat{B})&=\elbows{\phi|\otimes\langle\psi | \hat{A}^2\otimes \hat{B}^2 |\phi\rangle\otimes|\psi}-\elbows{\hat{A}}_\phi ^2\elbows{\hat{B}}_\psi ^2 \nonumber\\
   &=\langle\phi|\otimes\langle\psi | \left(\hat{A}^2|\phi\rangle\otimes \hat{B}^2|\psi\rangle\right) -\paren{\elbows{\hat{A}}_\phi\elbows{\hat{B}}_\psi}^2 \nonumber\\
   &=\elbows{\hat{A}^2}_\phi\elbows{\hat{B}^2}_\psi -\paren{\elbows{\hat{A}}_\phi\elbows{\hat{B}}_\psi}^2 .
\end{align}
If $\ket{\alpha}$ is entangled with Schmidt decomposition $\ket{\alpha} =\sum\lambda _i\ket{\phi _i}\otimes\ket{\psi _i}$, we obtain from \eqref{eq7} that
\begin{align}            
\label{eq8}
\Delta _\alpha (\hat{A}\otimes \hat{B})&=\elbows{\hat{A}^2\otimes \hat{B}^2}_\alpha -\elbows{\hat{A}\otimes \hat{B}}_\alpha ^2 \nonumber\\
   &=\sum _{i,j}\lambda _i\lambda _j\elbows{\phi _i|\hat{A}^2|\phi _j}\elbows{\psi _i|\hat{B}^2|\psi _j}
   -\paren{\sum _{i,j}\lambda _i\lambda _j\elbows{\phi _i|\hat{A}|\phi _j}\elbows{\psi _i|\hat{B}|\psi _j}}^2 .
\end{align}
If $\ket{\alpha}$ is maximally entangled, this reduces to
\begin{align}
\Delta _\alpha (\hat{A}\otimes \hat{B})
    =\frac{1}{n} \sum _{i,j}\elbows{\phi _i|\hat{A}^2|\phi _j}\elbows{\psi _i|\hat{B}^2|\psi _j}
    -\frac{1}{n^2}\paren{\sum _{i,j}\elbows{\phi _i|\hat{A}|\phi _j}\elbows{\psi _i|\hat{B}|\psi _j}}^2 .
\end{align}

As expected, in the separable case we have $\elbows{\hat{A}\otimes \hat{I}_2}_\alpha =\elbows{\hat{A}}_\phi$. However, when $\ket{\alpha}$ is entangled, we have by \eqref{eq7} that
\begin{equation}
\elbows{\hat{A}\otimes \hat{I}_2}_\alpha =\sum _i\lambda _i^2\elbows{\hat{A}}_{\phi _i}
\end{equation}
and when $\ket{\alpha}$ is maximally entangled, $\elbows{\hat{A}\otimes \hat{I}_2}_\alpha =\tfrac{1}{n}\sum\limits _i\elbows{\hat{A}}_{\phi _i}$. Similarly, when $\ket{\alpha}$ is separable, we have
\begin{equation}
\Delta _\alpha (\hat{A}\otimes \hat{I}_2)=\Delta _\phi (\hat{A})
\end{equation}
and when $\ket{\alpha}$ is entangled, we have by \eqref{eq8} that
\begin{equation}
\Delta _\alpha (\hat{A}\otimes \hat{I}_2)=\sum _i\lambda _i^2\elbows{\hat{A}^2}_{\phi _i}-\paren{\sum\lambda _i^2\elbows{\hat{A}}_{\phi _i}}^2 .
\end{equation}
If $\ket{\alpha}$ is maximally entangled, this becomes
\begin{equation}
\Delta _\alpha (\hat{A}\otimes \hat{I}_2)=\frac{1}{n}\sum _i\elbows{\hat{A}^2}_{\phi _i}-\frac{1}{n^2}\paren{\sum\elbows{\hat{A}}_{\phi _i}}^2 .
\end{equation}
Similar equations hold for $\hat{I}_1\otimes \hat{B}$.

For separable $\ket{\alpha}\in H_1\otimes H_2$, from \eqref{eq:94} it follows that $\rmcor _\alpha (\hat{A}\otimes \hat{I}_2,\hat{I}_1\otimes \hat{B})=0$, hence $\hat{A}\otimes \hat{I}_2$, $\hat{I}_1\otimes \hat{B}$ are uncorrelated.
Also, $\sqbrac{\hat{A}\otimes \hat{I}_2,\hat{I}_1\otimes \hat{B}}=0$ so no information is given by the uncertainty principle.

For the observable $\cscript =\brac{\hat{A}_x\otimes \hat{B}_y\colon (x,y)\in\Omega _\ascript\times\Omega _\bscript}$ with stochastic operator $\ctilde =\sum _{x,y}xy\hat{A}_x\otimes \hat{B}_y$ we have
\begin{equation}
\elbows{\cscript}_\alpha =\elbows{\ctilde\,}_\alpha =\elbows{\alpha |\sum _{x,y}xy\hat{A}_x\otimes \hat{B}_y |\alpha}=\sum _{x,y}xy\elbows{\alpha |\hat{A}_x\otimes \hat{B}_y|\alpha} .
\end{equation}
If $\ket{\alpha} =\sum\lambda _i\ket{\phi _i}\otimes\ket{\psi _i}$ is entangled, we obtain from \eqref{eq7} that
\begin{equation}
\elbows{\cscript}_\alpha =\elbows{\atilde\otimes\btilde\,}_\alpha
   =\sum _{i,j}\lambda _i\lambda _j\elbows{\phi _i|\atilde|\phi _j}\elbows{\psi _i|\btilde|\psi _j} .
\end{equation}
As before, when $\ket{\alpha} =\ket{\phi}\otimes\ket{\psi}$ is separable this becomes $\elbows{\cscript}_\alpha =\elbows{\ascript}_\phi\elbows{\bscript}_\psi$ and when $\ket{\alpha}$ is maximally entangled
\begin{equation}
\elbows{\cscript}_\alpha =\frac{1}{n}\sum _{i,j}\elbows{\phi _i|\atilde|\phi _j}\elbows{\psi _i|\btilde|\psi _j} .
\end{equation}

So far, we did not consider interactions because $\sqbrac{\hat{A}\otimes \hat{I}_2,\hat{I}_1\otimes \hat{B}}=0$. We now consider interactions on $H_1\otimes H_2$. Let $\hat{A},\hat{C}\in\lscript _S(H_1)$,
$\hat{B},\hat{D}\in\lscript _S(H_2)$ so that $\hat{A}\otimes \hat{B},\hat{C}\otimes \hat{D}\in\lscript _S(H_1\otimes H_2)$.
The interaction statistics are given by
\begin{align}            
\label{eq9}
\rmcor _\alpha (\hat{A}\otimes \hat{B},\hat{C}\otimes \hat{D})&=\elbows{\alpha |(\hat{A}\otimes \hat{B})(\hat{C}\otimes \hat{D})|\alpha}
-\elbows{\alpha |\hat{A}\otimes \hat{B}|\alpha}\elbows{\alpha |\hat{C}\otimes \hat{D}|\alpha}\notag\\
    &=\elbows{\alpha |\hat{A}\hat{C}\otimes \hat{B}\hat{D}|\alpha}-\elbows{\hat{A}\otimes \hat{B}}_\alpha\elbows{\hat{C}\otimes \hat{D}}_\alpha .
\end{align}
If $\ket{\alpha}$ has Schmidt decomposition $\ket{\alpha} =\sum _i\lambda _i\ket{\phi _i}\otimes\ket{\psi _i}$, then
\begin{align}
\rmcor _\alpha&(\hat{A}\otimes \hat{B},\hat{C}\otimes \hat{D})=\sum _i\lambda _i\langle\phi _i|\otimes\langle\psi _i|(\hat{A}\hat{C}\otimes \hat{B}\hat{D})\sum _j\lambda _j|\phi _j\rangle\otimes|\psi _j\rangle \nonumber\\
   &-\paren{\sum _{i,j}\lambda _i\lambda _j\elbows{\phi _i|\hat{A}|\phi _j}\elbows{\psi _i|\hat{B}|\psi _j}}\paren{\sum _{i,j}\lambda _i\lambda _j\elbows{\phi _i|\hat{C}|\phi _j}\elbows{\psi _i|\hat{D}|\psi _j}} \nonumber\\
   &=\sum _{i,j}\lambda _i\lambda _j\elbows{\phi _i|\hat{A}\hat{C}|\phi _j}\elbows{\psi _i|\hat{B}\hat{D}|\psi _j} \nonumber\\
   &\quad -\sum _{i,j,r,s}\lambda _i\lambda_j\lambda _r\lambda _s\elbows{\phi _i| \hat{A}|\phi _j}\elbows{\psi _i|\hat{B}|\psi _j}\elbows{\phi _r|\hat{C}|\phi _s}\elbows{\psi _r|\hat{D}|\psi _s} .
\end{align}
In the two extreme cases, when $\ket{\alpha} =\ket{\phi}\otimes\ket{\psi}$ is separable, we obtain
\begin{equation}
\rmcor _\alpha (\hat{A}\otimes \hat{B},\hat{C}\otimes \hat{D})=\elbows{\phi |\hat{A}\hat{C}|\phi}\elbows{\psi |\hat{B}\hat{D}|\psi}-\elbows{\hat{A}}_\phi\elbows{\hat{C}}_\phi\elbows{\hat{B}}_\psi\elbows{\hat{D}}_\psi
\end{equation}
and when $\ket{\alpha}$ is maximally entangled, we have
\begin{align}
\rmcor _\alpha (\hat{A}\otimes \hat{B},\hat{C}\otimes \hat{D})&=\frac{1}{n}\sum _{i,j}\elbows{\phi _i|\hat{A}\hat{C}|\phi _j}\elbows{\psi _i|\hat{B}\hat{D}|\psi _j} \nonumber\\
   &\quad -\frac{1}{n^2}\sum _{i,j,r,s}\elbows{\phi _i|\hat{A}|\phi _j}\elbows{\psi _i|\hat{B}|\psi _j}\elbows{\phi _r|\hat{C}|\phi _s}\elbows{\psi _r|\hat{D}|\psi _s} .
\end{align}
A particularly simple case is $\hat{B}=\hat{D}=\hat{I}_2$, we have observable operators $\hat{A}\otimes \hat{I}_2$, $\hat{C}\otimes \hat{I}_2$ and we still have interaction if $\sqbrac{\hat{A},\hat{C}}\ne 0$. In this case, if
$\ket{\alpha} =\sum\lambda _i\ket{\phi _i}\otimes\ket{\psi _i}$, then
\begin{equation}
\rmcor _\alpha (\hat{A}\otimes \hat{I}_2,\hat{C}\otimes \hat{I}_2)=\sum _i\lambda _i^2\elbows{\phi _i|\hat{A}\hat{C}|\phi _i}-\sum _{i,r}\lambda _i^2\lambda _r^2\elbows{\hat{A}}_{\phi _i}\elbows{\hat{C}}_{\phi _r} .
\end{equation}
For the two extreme cases, when $\ket{\alpha} =\ket{\phi}\otimes\ket{\psi}$ is separable, we have
\begin{equation}
\rmcor _\alpha (\hat{A}\otimes \hat{I}_2,\hat{C}\otimes \hat{I}_2)=\elbows{\phi |\hat{A}\hat{C}|\phi}-\elbows{\hat{A}}_\phi\elbows{\hat{C}}_\phi
\end{equation}
and when $\ket{\alpha}$ is maximally entangled, we obtain
\begin{equation}
\rmcor _\alpha (\hat{A}\otimes \hat{I}_2, \hat{C}\otimes \hat{I}_2)=\frac{1}{n}\sum _i\elbows{\phi _i|\hat{A}\hat{C}|\phi _i}-\frac{1}{n^2}\sum _{i,r}\elbows{\hat{A}}_{\phi _i}\elbows{\hat{C}}_{\phi _r} .
\end{equation}
It is straightforward to continue this discussion for interacting observables, $\hat{A}_x\otimes \hat{B}_y$, $\hat{C}_u\otimes \hat{D}_v$.

\section{Concluding remarks}

In this work, we have investigated four quantum entanglement
measures that are based on Schmidt decomposition. After normalization
of the measures, we have shown that partial order is
possible between some entanglement measures. In particular, we
have rigorously proved that the concurrence forms an upper bound on
the tangle and entanglement robustness, whereas the Schmidt number
forms a lower bound. The existing partial order was then utilized
to introduce the concept of relative sensitivity to quantum entanglement,
and with a minimal quantum toy model we have demonstrated
how concurrence can be used to sharply demarcate separable states
and Schmidt number to sharply demarcate maximally entangled states.

Each of the four entanglement measures could be computed using an explicit formula that is based on the Schmidt coefficients.
Performing singular value decomposition, however, is a computationally expensive task \cite{Ford2014}. Fortunately, a mathematical workaround proposed by Gudder \cite{Gudder2020a,Gudder2020b} could be utilized for those entanglement measures that require the sum of the fourth powers of the Schmidt coefficients, $\sum_i \lambda_i^4$. In the latter case, one can just compute the trace $\textrm{Tr}(\hat{C}\hat{C}^\dagger)^2$ obtained from the complex coefficient matrix $\hat{C}$, whose reshaping gives the bipartite quantum state vector in the basis $|i\rangle\otimes|j\rangle$, namely, $|\psi\rangle=\textrm{res} (\hat{C} )$ \cite{Miszczak2011}.
This allows for fast evaluation of the concurrence and the Schmidt number, which provide the upper and lower bounds on the amount of available quantum entanglement, respectively.

State-dependent entanglement statistics provides an alternative, but equivalent, theoretical characterization of separable/entangled states in terms of absence/presence of measurable correlations between all/some local quantum observables. Noteworthy, Schmidt decomposition features prominently in the calculation of expectation values, variances, covariances and correlations between quantum observables, which obey the uncertainty principle.

Quantum entanglement is a precious physical resource that allows quantum devices to outperform their classical counterparts in terms of speed and efficiency.
Therefore, the presented theorems with regard to the four entanglement measures: concurrence, tangle, entanglement robustness and Schmidt number, could be useful in practical quantum applications that require careful monitoring and utilization of the available quantum resources.\\

\end{document}